\documentclass{article}

\usepackage{graphicx} 

\usepackage{amssymb,amsmath}
\usepackage{amsthm}
\usepackage{url,hyperref}
\usepackage[capitalise]{cleveref}
\usepackage{todonotes}
\usepackage{algorithm}
\usepackage[noend]{algorithmic}
\usepackage[margin=2cm]{geometry}

\newtheorem{theorem}{Theorem}[section]
\newtheorem{proposition}[theorem]{Proposition}
\newtheorem{lemma}[theorem]{Lemma} 
\newtheorem{corollary}[theorem]{Corollary}

\theoremstyle{definition}
\newtheorem{example}[theorem]{Example}
\newtheorem{claim}{Claim}
\newtheorem{exercise}{Exercise}[section]
\newenvironment{exercise*}
{\exercise}
{\endexercise}

\newtheorem{remark}[theorem]{Remark}

\newtheorem{definition}[theorem]{Definition}

\newcommand{\Rb}{\operatorname{Rb}}
\newcommand{\Id}{\textrm{Id}}
\newcommand{\R}{\mathbb{R}}
\newcommand{\Pp}{\mathcal{P}}
\newcommand{\Q}{\mathcal{Q}}
\newcommand{\I}{\mathcal{I}}
\newcommand{\sse}{\subseteq}

\usepackage[noadjust]{cite}

\usepackage{tikz-cd}

\title{Reeb Graph of Sample Thickenings}
\author{H\aa vard Bakke Bjerkevik\thanks{Department of Mathematics \& Statistics, University at Albany, SUNY, USA; hbjerkevik@albany.edu}\quad\quad Nello Blaser\thanks{Department of Informatics, University of Bergen, Norway; nello.blaser@uib.no}\quad\quad Lars M. Salbu\thanks{Department of Informatics, University of Bergen, Norway; lars.salbu@uib.no}}

\begin{document}

\maketitle
\begin{abstract}
We consider the Reeb graph of a thickening of points sampled from an unknown space. Our main contribution is a framework to transfer reconstruction results similar to the well-known work of Niyogi, Smale, and Weinberger to the setting of Reeb graphs. To this end, we first generalize and study the interleaving distances for Reeb graphs. We find that many of the results previously established for constructible spaces also hold for general topological spaces. We use this to show that under certain conditions for topological spaces with real-valued Lipschitz maps, the Reeb graph of a sample thickening approximates the Reeb graph of the underlying space. Finally, we provide an algorithm for computing the Reeb graph of a sample thickening. 
\end{abstract}

\section{Introduction}

Collecting data through experiments typically does not give complete information about the system. Instead we get a \emph{finite} set of data points sampled from a larger space, and we want to study properties of this underlying space. \emph{Geometric reconstruction} concerns the problem of recovering topological information, like homology/homotopy groups \cite{Chazal2007,Fasy2022,Brun2023,chazal2008towards} or even the homotopy type \cite{Niyogi2008,attali2024tight,attali2013,wang2018,Fasy2022,kim2020,adams2019}, of an unknown space by considering a finite set of sampled points.
Assumptions on the underlying space are needed to ensure reconstruction, usually based on geometric properties like the (local) ($\mu-$)reach \cite{federer1959,Amenta1999,Chazal2009}, distortion \cite{Gromov1978}, convexity radius \cite{Hausmann1996} or weak feature size \cite{Chazal2007}. Additionally, the samples need to be dense and well-distributed, and are only sometimes allowed sample noise (e.g.\ \cite{attali2024tight}).

A classical way of describing shapes is with Reeb graphs \cite{biasotti2008}. Given an $\mathbb{R}$-space, namely a topological space \(X\) with a continuous function \(f: X \to \mathbb{R}\), the Reeb graph \(\Rb(X, f)\) is the quotient space constructed from $X$ by identifying points in the same connected component of level sets $f^{-1}(a)$ for $a\in\mathbb{R}$. Introduced by Reeb in 1946 \cite{Reeb1946}, it has found applications in diverse fields, ranging from computer graphics (see survey \cite{biasotti2008}) to neuroscience (e.g.\ \cite{Shailja2024,Shi2014}).

We consider the problem of approximating the Reeb graph of an unknown space from a sample and give a framework to transfer reconstruction results to this setting. In particular, we look at results that recover the homotopy type by constructing a larger space that deformation retracts to the unknown space, typically by thickening the point cloud \cite{Niyogi2008,attali2024tight,wang2018}.

To measure approximation quality, we need to compare Reeb graphs. Reeb distances include the bottleneck \cite{Cohen-Steiner2009},  interleaving \cite{de_Silva_2016,bauer_2015}, functional distortion and contortion distances \cite{bauer_2014,bauer_2022}, the Reeb radius \cite{curry2024}, and the universal distance \cite{bauer_2021}. A lot of previous work compares these distances \cite{bauer_2014,bauer_2015,bauer_2022,bollen2022,carrière2017} which are often defined for special Reeb graphs. For general $\mathbb{R}$-spaces, we use the interleaving distance with connected components (also discussed in \cite{de_Silva_2016} where they mainly consider path components). This distance has a general Reeb stability result (\cref{prop: Reeb_stability}), that yields approximation results without any extra assumption. With additional assumptions, our methods also works for other distances (\cref{remark: approximation_for_other_distances}).

Prior work on approximating Reeb graphs from samples includes that the Reeb graph of the Vietoris-Rips complex of a dense sample of a smooth compact manifold $M$ can approximate the Reeb graph $\Rb(M,f)$ where $f$ is level-set-tame Lipschitz \cite[Thm.\ 4.7]{Dey2013}. Mapper \cite{singh2007Mapper} can be seen as a discretized approximation of the Reeb graph. In particular, in \cite[Cor.\ 6]{munch2016} and \cite[Cor.\ 1]{Brown2021} they show that the geometric and enhanced Mapper, respectively, are close in interleaving distance to the Reeb graph for constructable $\mathbb{R}$-spaces. Moreover, in \cite[Thm.\ 7]{carriére2018} they show that Mapper  is close in bottleneck distance to the Reeb graph of Morse-type functions on spaces with positive reach and convexity radius.

On the computational front, an early contribution was \cite{Shinagawa1991} where they found the Reeb graph of Morse functions on triangulated 2-manifolds in $O(n^2)$, where $n$ is the number of triangles. It was later improved to $O(n \log n)$ in \cite{Cole-McLaughlin2004}. For the more general case of PL functions on simplicial complexes, randomized \cite{harvey2010} and later deterministic \cite{Parsa2013} $O(m \log m)$ algorithms have been suggested, where $m$ is the size of the simplicial $2$-skeleton.

\subsection{Contributions}
Our contributions are as follows: 
\begin{enumerate}
    \item In \cref{prop: Reeb_stability} we show that
    for continuous functions $f_1,f_2:X\to\mathbb{R}$, the interleaving distance between Reeb graphs \(\Rb(X,f_1)\) and \(\Rb(X,f_2)\) is bounded by \(\|f_1-f_2\|_\infty\). This generalizes results from \cite{de_Silva_2016}, where this is shown for constructible $\mathbb{R}$ spaces, to general $\mathbb{R}$-spaces. This result is interesting in itself, and necessary to prove subsequent results. 
    \item Our main result is a template for creating Reeb approximation results (\cref{thm: reconstruct-ideal}). Given a continuous function $g:Y\to\mathbb{R}$ and a deformation retraction $H:Y\times [0,1]\to Y$ of $X\subseteq Y$, such that $h_1\circ h_t = h_1$ for all $0\leq t\leq 1$, we found that the interleaving distance \(d_\mathrm{I}(\Rb(Y,g),\Rb(X,g|_X))\) is bounded by \(\|g-g \circ h_1\|_\infty\).
    \item As a direct consequence of the above result together with reconstruction results a la Niyogi, Smale, and Weinberger \cite{Niyogi2008}, we find for example that the Reeb graph of a closed Euclidean subset $X\subseteq\mathbb{R}^d$ with positive reach \(\tau(X)\) can be approximated from the Reeb graph of a sample \(A\subseteq X^\delta\) under some conditions (\cref{cor: closed_set_approximation1,cor: closed_set_approximation2,cor: closed_set_approximation3}). Similar results also hold for closed subsets of Riemannian manifolds (\cref{sec: riemann}). 
    \item An algorithm computing the Reeb graph of an \(\varepsilon\)-thickening of a set of points \(A \subset \mathbb{R}^d\) for linear functions \(f \colon \mathbb{R}^d \to \R\). Given \(n=|A|\) points and \(t \le n^2\) overlapping \(\varepsilon\)-balls, the algorithm runs in $O(n(n+t) \alpha(n))$ time, where $\alpha(n)$ is the inverse Ackermann function. 
\end{enumerate}

In terms of $n$, our algorithm runs in $O(n^3\alpha(n))$ in the worst case $t=\Theta(n^2)$.
This is not directly comparable to the algorithms mentioned above, since they take as input triangulated manifolds or simplicial complexes, while we take as input only a set of points.
Still, we observe that these algorithms run in $O(n^3\log n)$ or worse, where $n$ is the number of $0$-simplices.
Our algorithm compares favorably to this, since we replace a logarithmic factor with $\alpha(n)$, and we do not have to spend time to explicitly compute a triangulation or simplicial complex.

The paper is structured as follows. \cref{sec: reebfunctor} introduces \(\mathbb{R}\)-spaces and Reeb functors, and shows that $\mathbb{R}$-spaces and their Reeb graphs have the same number of connected components. 
In \cref{sec: interleaving}, we generalize the interleaving distance to all \(\mathbb{R}\)-spaces, concluding in a stability result (Contribution~1). In \cref{sec: reeb_approximation_framework}, we present our Reeb approximation template (Contribution~2) and in \cref{sec: reconstruction} we apply it to closed Euclidean subsets and closed subsets of Riemannian manifolds (Contribution~3). 
\cref{sec: computation} provides an algorithm to compute Reeb graphs from thickenings and its analysis (Contribution~4).

\subsection{Acknowledgments}

This work was in part funded by the European Union, GA\#101126560; Bergen research and training program for future AI leaders across the disciplines, LEAD AI. 
Part of this paper is based upon work supported by the Swedish Research Council under grant no. 2021-06594 while some of the authors were in residence at Institut Mittag-Leffler in Djursholm, Sweden during the summer of 2025.

\section{The Reeb Functor}\label{sec: reebfunctor}
In this section we give the necessary background, defining $\mathbb{R}$-spaces and their Reeb graphs. We work in the categorical setting, like in \cite{de_Silva_2016}.

An \textbf{$\mathbb{R}$-space} is a pair $(X,f)$ consisting of a topological space $X$ together with a continuous function $f:X\to\mathbb{R}$ \cite{de_Silva_2016}. We call $X$ the \textbf{underlying space} of the $\mathbb{R}$-space $(X,f)$. A \textbf{morphism} of $\mathbb{R}$-spaces $G:(Y,g)\to(X,f)$ is a function-preserving continuous function on underlying spaces, i.e.\ it is a continuous map $G:Y\to X$ such that $g=f\circ G$. We denote the \textbf{category of $\mathbb{R}$-spaces} by $\mathbf{Top}/\mathbb{R}$.

The \textbf{level-set} of an $\mathbb{R}$-space $(X,f)$ at level $a\in\mathbb{R}$ is the preimage 
    $f^{-1}(a)=\{x\in X\,|\, f(x)=a\}\subseteq X$
with the subspace topology. We define an equivalence relation $\sim_f$ on the underlying space $X$ by saying $x\sim_f y$ if and only if $x$ and $y$ are in the same connected component of the same level-set $f^{-1}(a)$. We denote the equivalence class of $x$ by $[x]_f$. The quotient space $X/_{\sim_f}$ is called the \textbf{Reeb graph} of $f$ \cite[Thm.\ 1]{Reeb1946}. Note that
\begin{equation}\label{eq: equivalent_same_value}
    \textrm{$f(x)=f(y)$ whenever $x\sim_f y$,}    
\end{equation}
so $f$ induces a continuous map $\Tilde{f} : X/_{\sim_f}\to\mathbb{R}$ defined by $[x]_f\mapsto f(x)$. In particular, the quotient map $q_f:X\to X/_{\sim_f}$ defines a morphism of $\mathbb{R}$-spaces $(X,f)\to (X/_{\sim_f},\Tilde{f})$. 

Let $G:(Y,g)\to (X,f)$ be a morphism of $\mathbb{R}$-spaces. If $g(y)=a$, then $f(G(y))=a$ so $G(y)$ is in $f^{-1}(a)$. Furthermore, since continuous maps send connected sets to connected sets \cite[Thm.\ 23.5]{Munkres_2014}, we get that if $y$ and $y'$ are in the same connected component of $g^{-1}(a)$, then $G(y)$ and $G(y')$ are in the same connected component of $f^{-1}(a)$. In particular, the mapping $[y]_g\mapsto [G(y)]_f$ defines a continuous map on quotient spaces $\Tilde{G}:Y/_{\sim_g}\to X/_{\sim_f}$ such that $\Tilde{g}=\Tilde{f}\circ\Tilde{G}$. Following \cite[Sec.\ 2.4]{de_Silva_2016}, we define the \textbf{Reeb functor} 
$\Rb:\mathbf{Top}/\mathbb{R}\to\mathbf{Top}/\mathbb{R}$,
sending object $(X,f)$ to $(X_{\sim_f},\Tilde{f})$ and morphism $G$ to $\Tilde{G}$. Abusing notation, we often write $\Rb(X,f)$ for the underlying space $X/_{\sim_f}$, leaving the map $\Tilde{f}$ implicit.

\begin{remark}
    Reeb graphs are not generally (topological) graphs \cite[Ex.\ 2.4]{Gelbukh2024}, but for some important classes of $\mathbb{R}$-spaces they are. For example, the Reeb graph $\Rb(X,f)$ is a topological graph under the following conditions:
    \begin{itemize}
        \item For a compact manifold $X$ with a Morse function $f$ \cite[Thm.\ 1]{Reeb1946}.
        \item For a connected, compact and triangulable space $X$ with $f$ continuous \cite[Def.\ 2.5]{bauer_2015}.
        \item For a constructable $\mathbb{R}$-space $(X,f)$ \cite[Sec.\ 2]{de_Silva_2016}.
    \end{itemize}
\end{remark}

In the following lemma, we see that the Reeb graph of the preimage $f^{-1}(U)$ is the same as the preimage $\Tilde{f}^{-1}(U)$ of the induced map $\Tilde{f}:\Rb(X,f)\to\mathbb{R}$ for open subsets $U\subseteq\mathbb{R}$.

\begin{lemma}\label{lemma: subspace-quotient-lemma}
    If $f:X\to \mathbb{R}$ is a continuous map and $U\subseteq \mathbb{R}$ is an open subset, then
    \begin{equation*}
        \Rb(f^{-1}(U),f)=q_f(f^{-1}(U))=\Tilde{f}^{-1}(U).
    \end{equation*}
\end{lemma}
\begin{proof}
    We write $A= f^{-1}(U)$. From \eqref{eq: equivalent_same_value} we get that 
    \begin{equation}\label{eq: A_saturated}
        \textrm{$x\in A$ and $x\sim_f y$} \implies y\in A.
    \end{equation}
    So, $[x]_f=[x]_{f|_A}$ for all $x\in A$. Furthermore, since $\Tilde{f}([x]_f)=f(x)$, all three spaces in question contain the same elements, namely the equivalence classes 
     \begin{equation*}
        [x]_f=\{y\,|\,\textrm{$x$ and $y$ are in the same connected component of $f^{-1}(f(x))$}\},
    \end{equation*}
    where $x$ is in $A$. To show that the topologies agree, we note that restricting a quotient map to an open subset satisfying \eqref{eq: A_saturated} gives a quotient map \cite[Thm.\ 22.1(1)]{Munkres_2014}. In particular, we get $q_{(f|_A)}=(q_f)|_A$, and the spaces $\Rb(A,f)$ and $q_f(A)$ have the same topology. The result now follows, as both $q_f(A)$ and $\Tilde{f}^{-1}(U)$ have the subspace topology and are equal as sets.
\end{proof}

To prove that the Reeb graph functor preserves connected components we first show that the image of connected components under the Reeb quotient map are closed and connected.

\begin{lemma}\label{lemma: component_lemma 0}
    Let $f:X\to\mathbb{R}$ be continuous with Reeb quotient map $q_f:X\to\Rb(X,f)$. If $B$ is a connected component of $X$, then $q_f(B)$ is closed and connected in $\Rb(X,f)$.
\end{lemma}
\begin{proof}
    The set $q_f(B)$ is connected as the continuous image of a connected set \cite[Thm.\ 23.5]{Munkres_2014}. Let $x\in X$ be such that $[x]_f\in q_f(B)$, so $x\sim_f b$ for some $b\in B$. In particular, $x$ and $b$ are in the same connected component $X$, i.e.\ $x$ is in $B$. Thus the preimage $q_f^{-1}(q_f(B))=B$ is closed \cite[Thm.\ V.3.2(3)]{dugundji1966}, and $q_f(B)$ is closed in the quotient topology.
\end{proof}
We get a similar result considering preimages of the Reeb quotient map:
\begin{lemma}[Generalizing {\cite[Prop.\ 3]{bauer_2021}}]\label{lemma: component lemma 1}
    Let $f:X\to\mathbb{R}$ be continuous with Reeb quotient map $q_f:X\to\Rb(X,f)$. If $K$ is closed and connected in $\Rb(X,f)$, $q_f^{-1}(K)$ is connected in X.
\end{lemma}
\begin{proof}
    We modify the proof of \cite[Prop.\ 3]{bauer_2021} slightly. Let $q_f^{-1}(K)=U\cup V$ where $U$ and $V$ are non-empty and closed in $q_f^{-1}(K)$, and thus closed in $X$ as $q_f^{-1}(K)$ is closed. Assume by contradiction that $U$ and $V$ are disjoint. 
    
    If $q_f(U)$ and $q_f(V)$ are \underline{not} disjoint, then there is a class $[x]_f\in q_f(U)\cup q_f(V)$, so there exist points $u\in U$ and $v\in V$ in the same connected component of $f^{-1}(f(x))$. This implies that $u$ and $v$ are in the same connected component $B$ of $X$. The image $q_f(B)$ is connected by \cref{lemma: component_lemma 0} and it intersects $K$ in $[x]_f$. Since $K$ is a connected component, we get $q_f(B)\subseteq K$ and thus $B\subseteq q_f^{-1}(K)=U\cap V$. Now $B$ is covered by the non-empty disjoint closed sets $(U\cap B)$ and $(V \cup B)$, contradicting connectivity of $B$.

    If $q_f(U)$ and $q_f(V)$ \underline{are} disjoint, note that $K$ is covered by non-empty disjoint sets, since
    \begin{equation*}
        K=q_f(q_f^{-1}(K))=q_f(U\cup V)=q_f(U)\cup q_f(V).
    \end{equation*}
    Consider $x\in X$ such that $[x]_f$ is in $q_f(U)\subseteq K$. In particular, $[x]_f$ is not in $q_f(V)$ and $x$ is not in $V$. Since $x\in q_f^{-1}(K)=U\cup V$, we conclude that $x\in U$, and so $q_f^{-1}(q_f(U))=U$, which is closed. Thus $q_f(U)$ is closed in the quotient topology, and by a similar argument so is $q_f(V)$. These sets form a non-empty closed disjoint cover of $K$, contradicting its connectivity. 
    
    We conclude that $U$ and $V$ must intersect, and that $q_f^{-1}(K)$ is connected.
\end{proof}

Let $C_0:\mathbf{Top}\to\mathbf{Set}$ be the functor sending topological spaces to their set of connected components. If $G:X\to Y$ is continuous and $B$ is a connected component of $X$, then $C_0(G)(B)$ is the connected component of $Y$ containing $B$ (well-defined by \cite[Thm.\ V.3.3]{dugundji1966}). 

\begin{proposition}\label{prop: Reeb_component}
    Let $(X,f)$ be an $\mathbb{R}$-space. The map $q_f:X\to \Rb(X,f)$ induces a bijection $C_0(X)\xrightarrow{\simeq} C_0(\Rb(X,f))$ that is natural in the sense that if $G:(Y,g)\to(X,f)$ is a morphism of $\mathbb{R}$-spaces, then we have a commuting diagram
    \begin{equation*}
        \begin{tikzcd}
            C_0(Y)\arrow[rr,"C_0(q_g)"]\arrow[rr,"\simeq",swap]\arrow[d,"C_0(G)"]&& C_0(\Rb(Y,g))\arrow[d,"C_0(\Rb(G))"]\\
            C_0(X)\arrow[rr,"C_0(q_{f})"]\arrow[rr,"\simeq",swap]&& C_0(\Rb(X,f)).
        \end{tikzcd}
    \end{equation*}
\end{proposition}
\begin{proof}   
    \textit{Surjectivity:} Consider a connected component $K$ of $\Rb(X,f)$. Let $[x]_f$ be a class in $K$ represented by $x$, and let $B$ be the connected component of $x$ in $X$. The set $q_f(B)$ is connected \cite[Thm.\ 23.5]{Munkres_2014}, intersecting the connected component $K$ in $[x]_f$, so $q_f(B)\subseteq K$.
    
    \textit{Injectivity:} Let $B_1$ and $B_2$ be disjoint connected components of $X$. Assume by contradiction that $q_f(B_1)$ and $q_f(B_2)$ are in the same connected component $K$ of $\Rb(X,f)$. By \cref{lemma: component lemma 1} the preimage $q_f^{-1}(K)$ is connected, and it intersects both $B_1$ and $B_2$. Thus $B_1=B_2$, which is a contradiction.

    \textit{Naturality:} Consider $G:(Y,g)\to (X,f)$ and let $y$ be any element in $Y$. We note that 
    \begin{equation*}
        q_f\circ G(y)=[G(y)]_f=\Rb(G)([y]_g)=\Rb(G)\circ q_g(y),
    \end{equation*}
    and naturality follows from functoriality of $C_0$.
\end{proof}

\section{Reeb Precosheaf and Interleaving} \label{sec: interleaving}

We now consider how to compare \(\mathbb{R}\)-spaces. In particular, we define the interleaving distance between $\mathbb{R}$-spaces, and show a Reeb stability result. This generalizes results from \cite{de_Silva_2016}.

Let $\mathbf{Int}$ be the category of open intervals of $\mathbb{R}$ with inclusions as morphisms. For $\delta\geq 0$, we have a \textbf{widening functor} $\omega_\delta:\mathbf{Int}\to\mathbf{Int}$ sending intervals $(a,b)$ to their \textbf{$\delta$-thickening} $(a-\delta,b+\delta)$. We say that two functors $C,D:\mathbf{Int}\to\mathcal{A}$ are \textbf{$\delta$-interleaved} if there are natural transformations $\Phi:C\to D \circ \omega_\delta$ and $\Psi:D\to C \circ \omega_\delta$ such that $\Psi_{\omega_\delta (I)}\circ \Phi_I = C(I\subseteq \omega_{2\delta}(I))$ and $\Phi_{\omega_\delta (I)}\circ \Psi_I = D(I\subseteq \omega_{2\delta}(I))$. In this case we say that the pair $(\Phi,\Psi)$ is a \textbf{$\delta$-interleaving} between $C$ and $D$. Note that $C$ and $D$ are $0$-interleaved if and only if they are naturally isomorphic. The \textbf{interleaving distance} \cite[Def.\ 4.2]{de_Silva_2016} between $C$ and $D$ is
\begin{equation*}
    d_\mathrm{I}(C,D)=\inf \{\delta\,|\,\textrm{ $C$ and $D$ are $\delta$-interleaved}\}.
\end{equation*}

\begin{proposition}\label{prop:functor-shrink-interleaving}
    Let $C,D:\mathbf{Int}\to\mathcal{A}$ and $H:\mathcal{A}\to\mathcal{B}$ be functors, then
    \begin{equation*}
        d_\mathrm{I}(H\circ C,H\circ D)\leq d_\mathrm{I}(C,D).
    \end{equation*}
\end{proposition}
\begin{proof}
    The proof is similar to that of {\cite[Prop.\ 3.6]{Bubenik_2014}} for persistence modules. Let $(\Phi,\Psi)$ be a $\delta$-interleaving between $C$ and $D$. Since $\Phi: C\to D \circ \omega_\delta$ is a natural transformation, then by functoriality so is $H\Phi:H\circ C\to H\circ D\circ \omega_\delta$ where $(H\Phi)_I=H(\Phi_I)$. Similar, we get a natural transformation $H\Psi:H\circ D\to H\circ C \circ \omega_\delta$. Using $\Psi_{\omega_\delta (I)}\circ \Phi_I = C(I\subseteq \omega_{2\delta}(I))$ and functoriality of $H$, we get
    \begin{equation*}
        H(\Psi_{\omega_\delta (I)})\circ H(\Phi_I) = H(\Psi_{\omega_\delta (I)}\circ \Phi_I) = H(C(I\subseteq \omega_{2\delta}(I))) = H\circ C(I\subseteq \omega_{2\delta}(I)),
    \end{equation*}
    and likewise $H(\Phi_{\omega_\delta (I)})\circ H(\Psi_I) = H\circ D(I\subseteq \omega_{2\delta}(I))$. So, $(H\Phi,H\Psi)$ is a $\delta$-interleaving between $H\circ C$ and $H \circ D$.
\end{proof}

For an $\mathbb{R}$-space $(X,f)$ define the \textbf{preimage functor} $\operatorname{Pre}(X,f):\mathbf{Int}\to \mathbf{Top}$ sending intervals $I$ to their preimage $f^{-1}(I)$ and inclusions $I\subseteq J$ to inclusion maps $f^{-1}(I)\hookrightarrow f^{-1}(J)$. The \textbf{Reeb precosheaf} (see \cite[Sec.\ 3.4]{de_Silva_2016}) denoted $\mathcal{D}(X,f)$ is the composition
\begin{equation*}
    \mathbf{Int}\xrightarrow{\operatorname{Pre}(X,f)}\mathbf{Top}\xrightarrow{C_0}\mathbf{Set},
\end{equation*}
so $\mathcal{D}(X,f)(I)$ is the set of connected components of the preimage $f^{-1}(I)$. 

\begin{remark}\label{remark: path-discussion}
    In \cite{de_Silva_2016,bauer_2015} they instead use path components, i.e. for $\mathbb{R}$-spaces $(X,f)$ define the \textbf{Reeb cosheaf} $\mathcal{C}(X,f)=\pi_0\circ \operatorname{Pre}(X,f):\mathbf{Int}\to\mathbf{Set}$, where $\pi_0$ sends spaces to their set of path components. The functor $\mathcal{C}(X,f)$ is a cosheaf \cite[Prop.\ 3.13]{de_Silva_2016}, while $\mathcal{D}(X,f)$ is generally not \cite[Ex.\ 3.18]{de_Silva_2016}. The two agree for \emph{constructable $\mathbb{R}$-spaces}, which are locally path-connected, so path components and connected components coincide (\cite[Thm.\ 25.5]{Munkres_2014}, \cite[Thm.\ 2.131]{Moller_notes}). In \cref{ex: Bjerkevik} we justify choosing $\mathcal{D}(X,f)$ by showing that the interleaving distance between $\mathcal{C}(\Rb(X,f))$ and $\mathcal{C}(\Rb(Y,g))$ can be both smaller and larger than the distance between $\mathcal{C}(X,f)$ and $\mathcal{C}(Y,g)$. \cref{lemma: Reeb_contracts} shows equality for connected components.
\end{remark}

\begin{theorem}
    \label{lemma: Reeb_contracts}
    Let $(X,f)$ and $(Y,g)$ be $\mathbb{R}$-spaces, then
    \begin{equation*}
        d_\mathrm{I}(\mathcal{D}(\Rb(Y,g)),\mathcal{D}(\Rb(X,f))) = d_\mathrm{I}(\mathcal{D}(Y,g),\mathcal{D}(X,f)).
    \end{equation*}
\end{theorem}
\begin{proof}
    First we show that $d_\mathrm{I}(\mathcal{D}(\Rb(Y,g)),\mathcal{D}(\Rb(X,f))) \leq d_\mathrm{I}(\mathcal{D}(Y,g),\mathcal{D}(X,f))$. 
    Let $(\Phi,\Psi)$ be a $\delta$-interleaving between $\mathcal{D}(Y,g)$ and $\mathcal{D}(X,f)$. In particular, we have maps $\Phi_{I}:\mathcal{D}(Y,g)(I)\to \mathcal{D}(X,f)(\omega_\delta(I))$ that are compatible with inclusions $I\subseteq J$ for open intervals $I$ and $J$. \cref{prop: Reeb_component} gives a natural bijection 
    \begin{equation*}
        \mathcal{D}(Y,g)(I)=C_0(g^{-1}(I))\xrightarrow[\cong]{C_0(q_g)}C_0(\Rb(g^{-1}(I),g))= \mathcal{D}(\Rb(Y,g))(I),
    \end{equation*}
    where the last equality is from \cref{lemma: subspace-quotient-lemma}. We compose natural transformations
    \begin{equation*}
         \mathcal{D}(\Rb(Y,g))\xrightarrow{C_0(q_g)^{-1}}\mathcal{D}(Y,g)\xrightarrow{\Phi} \mathcal{D}(X,f)\circ \omega_\delta \xrightarrow{C_0(q_f)}\mathcal{D}(\Rb(X,f))\circ \omega_\delta,
    \end{equation*}    
     giving a natural transformation $\phi:\mathcal{D}(\Rb(Y,g))\to\mathcal{D}(\Rb(X,f))\circ\omega_\delta$. Similarly we define $\psi=C_0(q_g)\circ\Psi\circ C_0(q_f)^{-1}:\mathcal{D}(\Rb(X,f))\to\mathcal{D}(\Rb(Y,g))\circ\omega_\delta$, and consider the composition
     \begin{equation*}
         \psi_{\omega_\delta I}\circ \phi_{I} = C_0(q_g)\circ\Psi_{\omega_\delta I}\circ C_0(q_f)^{-1}\circ  C_0(q_f)\circ\Phi_{I}\circ C_0(q_g)^{-1}.
     \end{equation*}
     Using the fact that $(\Phi,\Psi)$ is a $\delta$-interleaving, this reduces to the composition
     \begin{equation*}
        C_0(\Tilde{g}^{-1}(I))\xrightarrow{C_0(q_g)^{-1}} C_0(g^{-1}(I))\xrightarrow{C_0(\iota)}C_0(g^{-1}(\omega_\delta I))\xrightarrow{C_0(q_g)}C_0(\Tilde{g}^{-1}(\omega_\delta I)),
     \end{equation*}
     where $\iota:g^{-1}(I)\hookrightarrow g^{-1}(\omega_\delta I)$ is the inclusion. Consider $y\in g^{-1}(I)$. The component of $q_g(y)$ in $\Tilde{g}^{-1}(I)$ is sent by $C_0(\iota)\circ C_0(q_g)^{-1}$ to the component of $y$ in $g^{-1}(\omega_\delta I)$, which by $C_0(q_g)$ is sent to the component of $q_g(y)$ in $\Tilde{g}^{-1}(\omega_\delta I)$. Thus, the composition is induced by the inclusion $\Tilde{g}^{-1}(I)\hookrightarrow\Tilde{g}^{-1}(\omega_\delta I)$. The same is true for $\phi_{\omega_\delta I}\circ \psi_{I}$ by a symmetrical argument, so $(\phi,\psi)$ defines a $\delta$-interleaving.
     
     The converse follows from a similar argument. Let $(\phi,\psi)$ be a $\delta$-interleaving between $\mathcal{D}(\Rb(Y,g))$ and $\mathcal{D}(\Rb(X,f))$. In particular, we have maps $\phi_{I}:\mathcal{D}(\Rb(Y,g))(I)\to \mathcal{D}(\Rb(X,f))(\omega_\delta(I))$ that are compatible with inclusions $I\subseteq J$ for open intervals $I$ and $J$. From \cref{prop: Reeb_component} we get a composition of natural transformations
    \begin{equation*}
    \mathcal{D}(Y,g)\xrightarrow{C_0(q_g)}\mathcal{D}(\Rb(Y,g))\xrightarrow{\phi} \mathcal{D}(\Rb(X,f))\circ \omega_\delta \xrightarrow{C_0(q_f)^{-1}}\mathcal{D}(X,f)\circ \omega_\delta,
    \end{equation*}    
    defining a natural transformation $\Phi:\mathcal{D}(Y,g)\to\mathcal{D}(X,f)\circ\omega_\delta$. Similarly we define $\Psi=C_0(q_g)^{-1}\circ\psi\circ C_0(q_f):\mathcal{D}(X,f)\to\mathcal{D}(Y,g)\circ\omega_\delta$, and consider the composition
    \begin{equation*}
    \Psi_{\omega_\delta I}\circ \Phi_{I} = C_0(q_g)^{-1}\circ\psi_{\omega_\delta I}\circ C_0(q_f)\circ  C_0(q_f)^{-1}\circ\phi_{I}\circ C_0(q_g).
    \end{equation*}
    Using the fact that $(\phi,\psi)$ is a $\delta$-interleaving, this reduces to the composition
    \begin{equation*}
    C_0(g^{-1}(I))\xrightarrow{C_0(q_g)} C_0(\Tilde{g}^{-1}(I))\xrightarrow{C_0(\iota)}C_0(\Tilde{g}^{-1}(\omega_\delta I))\xrightarrow{C_0(q_g)^{-1}}C_0(g^{-1}(\omega_\delta I)),
    \end{equation*}
    where $\iota:\Tilde{g}^{-1}(I)\hookrightarrow \Tilde{g}^{-1}(\omega_\delta I)$ is the inclusion. Consider $y\in g^{-1}(I)$. It is sent to $[y]_g$ in $C_0(\Tilde{g}^{-1}(\omega_\delta I)$ by $C_0(\iota)\circ C_0(q_g)$, but this is represented by $y\in g^{-1}(\omega_\delta I)$. Thus, the composition is induced by the inclusion $g^{-1}(I)\hookrightarrow g^{-1}(\omega_\delta I)$. The same is true for $\Phi_{\omega_\delta I}\circ \Psi_{I}$ by an analogous argument.
\end{proof}

The \textbf{uniform distance} between real-valued functions $f,f':X\to\mathbb{R}$ is given by 
\begin{equation*}
    \|f-f'\|_\infty = \sup_{x\in X}|f(x)-f'(x)|.
\end{equation*}
This gives an upper bound for the interleaving distance both between the Reeb precosheaves and between the preimage functors of the functions:
\begin{lemma}[Similar to {\cite[Thm.\ 4.4(i)]{de_Silva_2016}}]\label{lemma: interleaving_bounded}
    Let $f_1,f_2:X\to\mathbb{R}$ be continuous. We have
    \begin{equation*}
        d_{\mathrm{I}}(\mathcal{D}(X,f_1),\mathcal{D}(X,f_2))\leq d_{\mathrm{I}}(\operatorname{Pre}(X,f_1),\operatorname{Pre}(X,f_2))\leq \|f_1-f_2\|_\infty.
    \end{equation*}
\end{lemma}
\begin{proof}
The first inequality follows from \cref{prop:functor-shrink-interleaving}. For the second, let $\delta=\|f_1-f_2\|_\infty$, implying that $f_1(x)-\delta\leq f_2(x)\leq f_1(x)+\delta$ for all $x$ in $X$. In particular, if $(a,b)$ is an interval and $x$ is in $f_1^{-1}(a,b)$, then $x$ is also in $f_2^{-1}(a-\delta,b+\delta)$. Thus we have inclusions
\begin{equation*}
    \operatorname{Pre}(X,f_1)(a,b)=f_1^{-1}(a,b)\hookrightarrow f_2^{-1}(a-\delta,b+\delta) = \operatorname{Pre}(X,f_2)\circ \omega_\delta(a,b)
\end{equation*}
defining a natural transformation $\operatorname{Pre}(X,f_1)\to\operatorname{Pre}(X,f_2)\circ \omega_\delta$. Together with the similarly defined $\operatorname{Pre}(X,f_2)\to\operatorname{Pre}(X,f_1)\circ \omega_\delta$, this gives a $\delta$-interleaving between $\operatorname{Pre}(X,f_1)$ and $\operatorname{Pre}(X,f_2)$ where the conditions on morphisms trivially hold as all maps are inclusions. 
\end{proof}

To ease on notation, we define the \textbf{interleaving distance} $d_\mathrm{I}$ between $\mathbb{R}$-spaces $(X,f)$ and $(Y,g)$ to be the interleaving distance between their Reeb precosheaves, i.e.\
\begin{equation*}
    d_\mathrm{I}((X,f),(Y,g)):=d_\mathrm{I}(\mathcal{D}(X,f),\mathcal{D}(Y,g)).
\end{equation*}
This simply extends the interleaving/Reeb distance of \cite[Sec.\ 2.2]{bauer_2015} and \cite[Def.\ 4.2]{de_Silva_2016} to include all $\mathbb{R}$-spaces, not only Reeb graphs. Combining \cref{lemma: Reeb_contracts,lemma: interleaving_bounded} we get the following stability, generalizing the result of {\cite[Thm.\ 4.4(ii)]{de_Silva_2016}} to non-constructable $\mathbb{R}$-spaces.

\begin{theorem}[Reeb Stability]\label{prop: Reeb_stability}
    If $f_1,f_2:X\to\mathbb{R}$ are continuous, then 
    \begin{equation*}
        d_{\mathrm{I}}(\Rb(X,f_1),\Rb(X,f_2))\leq \|f_1-f_2\|_\infty.
    \end{equation*}
\end{theorem}
We also need the following result to ensure Reeb approximation:

\begin{proposition}\label{prop: isostability}
   Let $(Y,g)$ be an $\mathbb{R}$-space. If there is an $\mathbb{R}$-space isomorphism $F:(X_1,f_1)\to(X_2,f_2)$, then
    \begin{equation*}
        d_\mathrm{I}((Y,g),(X_1,f_1))=d_\mathrm{I}((Y,g),(X_2,f_2)).
    \end{equation*}
\end{proposition}
\begin{proof}
    If $(\Phi,\Psi)$ is a $\delta$-interleaving between $\mathcal{D}(Y,g)$ and $\mathcal{D}(X_1,f_1)$, then the pair $(\mathcal{D}(F) \circ \Phi,\Psi \circ \mathcal{D}(F^{-1}))$ is a $\delta$-interleaving between $\mathcal{D}(Y,g)$ and $\mathcal{D}(X_2,f_2)$. Indeed, by functoriality they are natural transformations, and 
    \begin{equation*}
       \Psi_{\omega_\delta I} \circ \mathcal{D}(F^{-1}) \circ \mathcal{D}(F) \circ \Phi_I = \Psi_{\omega_\delta I}  \circ \Phi_I = \mathcal{D}(Y,g)(I\subseteq \omega_\delta I).
       \end{equation*}
    For the other composition, we note that 
    \begin{equation*}
        \mathcal{D}(F) \circ \Phi_{\omega_\delta I} \circ \Psi_{I} \circ \mathcal{D}(F^{-1}) = \mathcal{D}(F) \circ \mathcal{D}(\iota) \circ \mathcal{D}(F^{-1}) = \mathcal{D}(F\circ \iota\circ F^{-1})
    \end{equation*}
    where $\iota:f_1^{-1}(I)\hookrightarrow f_1^{-1}(\omega_\delta I)$ is the inclusion.
    In particular, if $x\in f_2^{-1}(I)$, then $F(\iota(F^{-1}(x)))=x$, and so $F\circ \iota\circ F^{-1}$ is the inclusion $f_2^{-1}(I)\hookrightarrow f_2^{-1}(\omega_\delta I)$. The result now follows by a symmetrical argument interchanging $1$'s and $2$'s.
\end{proof}

We end the section with an example showing that \cref{lemma: Reeb_contracts} does not hold for the Reeb cosheaf $\mathcal{C}(X,f)$ as defined in \cref{remark: path-discussion}. In fact, the two distances are not comparable.
\begin{example}[See also {\cite[3.18 and 3.19]{de_Silva_2016}}]\label{ex: Bjerkevik} 
    Consider the (closed) topologist's sine curve
    \begin{equation*}
        S = \left\{\left(x,\sin\frac{1}{x}\right)\,\middle|\, 0 < x\leq 1\right\}\cup \left\{\left(0,y\right)\,\middle|\, -1\leq y\leq 1\right\},
    \end{equation*}
    and the projection $p:S\to\mathbb{R}$ sending $(x,y)$ to $x$. The Reeb graph $(\Rb(S,p),\Tilde{p})$ is isomorphic to $([0,1],\iota)$ where $\iota:[0,1]\hookrightarrow\mathbb{R}$ is the inclusion \cite[3.18]{de_Silva_2016}. Furthermore, consider the disjoint union $T=\{0\}\sqcup (0,1]$ together with the map continuous $g:T\to \mathbb{R}$ where $g(t)=t$. Note that $t\sim_g t'$ in $T$ if and only if $t=t'$, so $(T,g)$ and $(\Rb(T,g),\Tilde{g})$ are isomorphic.

    There is a natural isomorphism $\Theta:\mathcal{C}(S,p)\to\mathcal{C}(T,g)$ where for intervals $I$ the map $\Theta_I:\pi_0(p^{-1}(I))\to \pi_0(g^{-1}(I))$ sends the path component $\{0\}\times [-1,1]$ to $\{0\}$ if $0\in I$ and sends $\{(x,\sin{1/x})\,|\, x\in (0,1]\cap I\}$ to the path component $(0,1]\cap I$ if $(0,1]\cap I\neq\emptyset$. In particular, $d_\mathrm{I}(\mathcal{C}(S,p),\mathcal{C}(T,g))=0$. 
    
    Using functoriality of $\mathcal{C}$ \cite[Sec.\ 3.4]{de_Silva_2016} we get that $d_\mathrm{I}(\mathcal{C}(\Rb(S,p)),\mathcal{C}(\Rb(T,g)))$ is equal to $d_\mathrm{I}(\mathcal{C}([0,1],\iota),\mathcal{C}(T,g))$ as their first arguments are isomorphic. Let $I$ be an interval. If $0\in I$, then $0\in\omega_\delta(I)$ for all $\delta\geq 0$. For $\mathcal{C}([0,1],\iota)$ and $\mathcal{C}(T,g))$ to be a $\delta$-interleaved, the map $\pi_0(g^{-1}(I))\to \pi_0(g^{-1}(\omega_{2\delta}(I)))$ induced by inclusion must factor through $\pi_0(\iota^{-1}(\omega_\delta(I)))$. However, both $g^{-1}(I)$ and $g^{-1}(\omega_{2\delta}I)$ have two path components that are mapped one-to-one by the inclusion, while $\iota^{-1}(\omega_\delta(I))$ has only one path component, making it impossible. Thus,  
    \begin{equation*}
        0 = d_\mathrm{I}(\mathcal{C}(S,p),\mathcal{C}(T,g)) < d_\mathrm{I}(\mathcal{C}(\Rb(S,p)),\mathcal{C}(\Rb(T,g))) = \infty.
    \end{equation*}
    Conversely, we also note that
    \begin{equation*}
        \infty = d_\mathrm{I}(\mathcal{C}(S,p),\mathcal{C}([0,1],\iota)) > d_\mathrm{I}(\mathcal{C}(\Rb(S,p)),\mathcal{C}(\Rb([0,1],\iota))) = 0.
    \end{equation*}
\end{example}

\section{Reeb Approximation Framework} \label{sec: reeb_approximation_framework}
Here we give the main result that allow us to approximate Reeb graphs using known reconstruction results. We define path deformation retractions, which preserve Reeb graphs in some sense. All reconstruction results we consider are path deformation retractions.

For a subspace $X\subseteq Y$, a \textbf{deformation retraction} of $X$ in $Y$ is a continuous map $H:Y\times [0,1]\to Y$ such that $H(y,0)=y$ and $H(y,1)\in X$ for all $y\in Y$ and $H(x,1)=x$ for all $x\in X$. In this case, $X$ is a \textbf{deformation retract} of $Y$. We write $h_t:Y\to Y$ for the continuous map $y\mapsto H(y,t)$ given $0\leq t\leq 1$. The map $H$ is \textbf{level-set preserving} with respect to a function $g:Y\to\mathbb{R}$ if $g(h_t(y)) = g(y)$ for all $(y,t)$ in $Y\times [0,1]$.

\begin{lemma} \label{lemma: path_lemma}
    If $H:Y\times [0,1]\to Y$ is a deformation retraction that is level-set preserving with respect to $g:Y\to\mathbb{R}$, then for $y\in Y$ there exists a path in $g^{-1}(g(y))$ from $y$ to $h_1(y)$.
\end{lemma}
\begin{proof}
    Let $y$ be any point in $Y$. Consider the path $\gamma_y:[0,1]\to Y$ defined by $\gamma_y(t)=h_t(y)$. This is a path from $h_0(y)=y$ to $h_1(y)$ and $g(\gamma_y(t))=g(h_t(y))=g(y)$ for all $0\leq t\leq 1$. Thus, the image $\textrm{Im} (\gamma_y)$ is a subset of $g^{-1}(g(y))$.
\end{proof}

A deformation retraction $H:Y\times [0,1]\to Y$ of $X$ in $Y$ is called a \textbf{path deformation retraction} if $h_1\circ h_t = h_1$ for all $0\leq t \leq 1$. In this case $X$ is a \textbf{path deformation retract} of $Y$. If $f:X\to \mathbb{R}$ is a continuous map, we get $f \circ h_1\circ h_t = f \circ h_1$, which is exactly the definition that $H$ is level-set preserving with respect to $f\circ h_1$:

\begin{lemma} \label{lemma: path-level-set}
    Let $f:X\to\mathbb{R}$ be a continuous function. A path deformation retraction $H:Y\times [0,1]\to Y$ of $X$ in $Y$ is level-set preserving with respect to \(f\circ h_1\). \qed
\end{lemma}

Consider a deformation retraction $H$ of $X$ in $Y$ together with a map $f:X\to \mathbb{R}$. Since $h_1(y)$ is in $X$ for all $y$ in $Y$, we can consider the function $h_1:Y\to X$ and the inclusion $\iota:X\hookrightarrow Y$. These maps induce morphisms of $\mathbb{R}$-spaces $h_1:(Y,f\circ h_1)\to (X,f)$ and $\iota:(X,f)\to(Y,f\circ h_1)$ where $h_1\circ\iota = \Id_X$. Reeb functoriality gives the commuting diagram:
\begin{equation*}
    \begin{tikzcd}
        (Y,f\circ h_1) \arrow[r,"h_1"]\arrow[d] &
        (X,f) \arrow[r,"\iota"]\arrow[d]  &
        (Y,f\circ h_1)\arrow[d]\arrow[r,"h_1"] &  
        (X,f) \arrow[d]
        \\
        \Rb(Y,f\circ h_1)\arrow[r,"\Rb(h_1)"] &
        \Rb(X,f)\arrow [r,"\Rb(\iota)"] &
        \Rb(Y,f\circ h_1) \arrow[r,"\Rb(h_1)"] &
        \Rb(X,f).
    \end{tikzcd}
\end{equation*}

\begin{proposition}\label{prop: Reeb-isomorphism}
    Let $f:X\to\mathbb{R}$ be a continuous function. If $H:Y\times [0,1]\to Y$ is a path deformation retraction of $X$ in $Y$, then $\Rb(h_1): \Rb(Y,f\circ h_1)\to \Rb(X,f)$ is an $\mathbb{R}$-space isomorphism with $\Rb(\iota)$ as its inverse.
\end{proposition}
\begin{proof}
    Since $h_1\circ \iota$ is the identity, by Reeb functoriality so is $\Rb(h_1)\circ \Rb(\iota)$. We show that $\Rb(\iota)\circ \Rb(h_1):\Rb(Y,f\circ h_1)\to \Rb(Y,f\circ h_1)$ mapping $[y]$ to $[h_1(y)]$ is the identity, by showing $[h_1(y)]=[y]$ for all $y\in Y$. By \cref{lemma: path-level-set}, $H$ is level-set preserving with respect to $f\circ h_1$. By \cref{lemma: path_lemma}, the points $y$ and $h_1(y)$ are in the same path component of $(f\circ h_1)^{-1}(f(h_1(y)))$. In particular, $y\sim_{f\circ h_1} h_1(y)$, and so $[y]=[h_1(y)]$. 
\end{proof}

The following is the main theorem that allow us to transfer geometric reconstruction results to the setting of Reeb graphs:

\begin{theorem}[Approximation Template]\label{thm: reconstruct-ideal}
    Let $g:Y\to\mathbb{R}$ be a continuous map. If $H:Y\times [0,1]\to Y$ is a path deformation retraction of $X$ in $Y$, then 
    \begin{equation*}
        d_\mathrm{I}(\Rb(Y,g),\Rb(X,g|_X)) \leq \|g-g \circ h_1\|_\infty.
    \end{equation*}
\end{theorem}
\begin{proof}
    By \cref{prop: Reeb-isomorphism}, the map $h_1$ induces an $\mathbb{R}$-space isomorphism on Reeb graphs $\Rb(h_1):\Rb(Y,g|_X\circ h_1)\to\Rb(X,g|_X)$. \cref{prop: isostability} states that isomorphisms does not change the interleaving distance. Thus we get the equality $d_\mathrm{I}(\Rb(Y,g),\Rb(X,g|_X))=d_\mathrm{I}(\Rb(Y,g),\Rb(Y,g\circ h_1))$ which is less than or equal to  $\| g - g\circ h_1\|_\infty$ by \cref{prop: Reeb_stability}.
\end{proof}

\begin{remark}\label{remark: approximation_for_other_distances}
    Note that \cref{thm: reconstruct-ideal} holds for any distance that satisfy the properties in \cref{prop: Reeb_stability} and \cref{prop: isostability}, and there are several distances where these properties hold for special cases of restricted $\mathbb{R}$-spaces. A survey on this topic can be found in \cite{bollen2022}, where they consider distances with \emph{Reeb stability} like in \cref{prop: Reeb_stability} \cite[Sec.\ 8.1]{bollen2022} and distances that are \emph{isomorphism indiscernible} \cite[Sec.\ 8.2]{bollen2022}, namely where Reeb graphs are isomorphic if and only if their distance is zero. This is stronger than our needed property of \cref{prop: isostability}.
\end{remark}

When reconstructing a space from a point cloud, we typically work in a metric space. For $k\geq 0$, we say that a function $g:M\to\mathbb{R}$ from a metric space $(M,d_M)$ is \textbf{$k$-Lipschitz} if $|g(p)-g(q)|\leq k \cdot d_M(p,q)$ for all $p,q\in M$.

\begin{corollary}\label{cor: ideal_lipschitz}
    Let $(M,d_M)$ be a metric space and let $H:M\times [0,1]\to M$ be a path deformation retraction of $X$ in $M$. If $\varepsilon,k>0$ such that $d(p,h_1(p)) < \varepsilon$ for all $p\in M$ and $g:M\to\mathbb{R}$ is $k$-Lipschitz, then 
    \begin{equation*}
        d_\mathrm{I}(\Rb(M,g),\Rb(X,g|_X)) < k\cdot\varepsilon.
    \end{equation*}
\end{corollary}
\begin{proof}
    From \cref{thm: reconstruct-ideal}, we get that $d_\mathrm{I}(\Rb(M,g),\Rb(X,g|_X)) \leq \|g - g\circ h_1\|_\infty$. Now,
    \begin{equation*}
        \|g - g\circ h_1\|_\infty = \sup_{p\in M} |g(p)-g(h_1(p))| \leq k \cdot \sup_{p\in M} d_M(p,h_1(p)) < k\cdot\varepsilon,
    \end{equation*}
    which concludes the proof.
\end{proof}

\section{Reeb Approximation from Reconstruction} \label{sec: reconstruction}
We now consider results in geometric reconstruction that recover the homotopy type of an unknown space from a point cloud sample. By applying the theory developed in \cref{sec: reeb_approximation_framework} we get results to approximate the Reeb graph of the space using only the sample.

\subsection{Closed Euclidean Subsets}
Here we consider closed Euclidean subsets \cite[Sec.\ 4]{attali2024tight}, starting with some general background by Federer \cite{federer1959}. A version for closed subsets of a Riemannian manifolds is shown in \cref{sec: riemann}.

Let $X\subseteq (M,d_M)$ be a metric subspace. For $r\geq 0$, the \textbf{$r$-thickening} $X^r$ of $X$ (in $M$ with respect to $d_M$) is the union of closed balls $B_r(x)$ of radius $r$ centered at points in $x\in X$,
\begin{equation*}
    X^r = \{p\in M\,|\, \exists x\in X \textrm{ with } d_M(x,p) \leq r\}.
\end{equation*}
We define the set $\operatorname{Unp}(X)$ of points in $M$ with unique nearest point in $X$ \cite[Def.\ 4.1]{federer1959}, so
\begin{equation*}
    \operatorname{Unp}(X)=\{p\in M\,|\, \exists x\in X \textrm{ s.t. }d_M(x,p) < d_M(x',p)\textrm{ for all } x'\in X\setminus\{x\}\}.
\end{equation*}
For a Euclidean subset $X\subseteq\mathbb{R}^d$ the \textbf{reach} $\tau(X)$ of $X$ is the biggest radius $r$ such that all points in $X^r$ have a unique nearest point in $X$ \cite[Def.\ 4.1]{federer1959}, namely
\begin{equation*}
    \tau(X) = \sup\{r\,|\, X^r\subseteq \operatorname{Unp}(X)\}.
\end{equation*}

\begin{theorem}[{\cite[Thm.\ 3]{attali2024tight}}]\label{thm: attali_thm3}
    Let $X\subseteq\mathbb{R}^d$ be a closed Euclidean subset with positive reach. Furthermore, let $0\leq \delta<\beta<\tau(X)$, $0<\alpha$ and $0\leq \varepsilon\leq\sqrt{(\beta-\delta)^2-(\beta-\alpha)^2}$. If $A\subseteq X^\delta$ and $X^\alpha\subseteq A^\varepsilon$, then $X$ is a strong deformation retract of $A^\varepsilon$ along the nearest point projection.
\end{theorem}

To be precise, the deformation retraction is the straight line homotopy $H:A^\varepsilon\times [0,1]\to A^\varepsilon$ sending $(u,t)$ to $(1-t)u + t\pi(u)$ where $\pi:A^\varepsilon\to X$ sends a point to its unique nearest point in $X$. In particular, for $u\in A^\varepsilon$, the point $h_t(u)=H(u,t)$ is on the straight line between $u$ and $\pi(u)$, so $d_\textrm{E}(u,\pi(u)) = d_\textrm{E}(u,h_t(u))+d_\textrm{E}(h_t(u),\pi(u))$, where $d_\textrm{E}$ denotes the Euclidean metric. Using the fact that $\pi(h_t(u))$ is the closest point to $h_t(u)$ in $X$, we get 
\begin{align*}
    d_\textrm{E}(u,\pi(h_t(u)))&\leq d_\textrm{E}(u,h_t(u)) + d_\textrm{E}(h_t(u),\pi(h_t(u)))\\ &\leq d_\textrm{E}(u,h_t(u)) + d_\textrm{E}(h_t(u),\pi(u)) \\
    &= d_\textrm{E}(u,\pi(u)).
\end{align*}
Since $\pi(u)$ is the unique nearest point to $u$, this means that $\pi(h_t(u))=\pi(u)$. Thus, $h_1\circ h_t = h_1$ for all $0\leq t\leq 1$, and $H$ is a path deformation retraction. Combining \cref{thm: reconstruct-ideal,thm: attali_thm3}, we get a Reeb graph approximation result.

\begin{corollary}\label{cor: closed_set_approximation1}
    Let $X\subseteq\mathbb{R}^d$ be a closed Euclidean subset with positive reach. Furthermore, let $0\leq \delta<\beta<\tau(X)$, $0<\alpha,k$ and $0\leq \varepsilon\leq\sqrt{(\beta-\delta)^2-(\beta-\alpha)^2}$. If $A\subseteq X^\delta$, $X^\alpha\subseteq A^\varepsilon$ and $g:A^\varepsilon\to\mathbb{R}$ is $k$-Lipschitz, then
         $d_\mathrm{I}(\Rb(A^\varepsilon,g),\Rb(X,g|_X)) < k(\varepsilon+\delta)$.
\end{corollary}
\begin{proof}
    By \cref{thm: attali_thm3}, we have that $X$ is a deformation retract of $A^\varepsilon$ by the straight line homotopy, which we just showed is a path deformation retraction. By definition, a point $u\in A^\varepsilon$ lies within distance $\varepsilon$ of some point in $A\subseteq X^\delta$, and thus within distance $(\varepsilon+\delta)$ of some point in $X$. Now, since $h_1(u)=\pi(u)$ is the nearest point in $X$ to $u$, we conclude that $d_\textrm{E}(u,h_1(u))<(\varepsilon+\delta)$. The result follows from \cref{cor: ideal_lipschitz}. 
\end{proof}
To get reconstruction of spaces in \cref{thm: attali_thm3}, we need the space to be well-behaved (closed Euclidean subspace with positive reach), and the sample must be dense, well-distributed ($X^\alpha\subseteq A^\varepsilon$) and not too noisy ($A\subseteq X^\delta$). For approximating Reeb graphs (\cref{cor: closed_set_approximation1}), we need an additional condition on the function ($k$-Lipschitz) ensuring that it is controlled.

Following a similar argument as in \cref{cor: closed_set_approximation1}, we can also transfer the following reconstruction result to a Reeb approximation result:
\begin{theorem}[{\cite[Prop.\ 5]{attali2024tight}}]\label{thm: attali_prop5}
    Let $X\subseteq\mathbb{R}^d$ be a closed Euclidean subset with positive reach and let $0<\beta<\tau(X)$. If $A\subseteq X^\delta$ and $X\subseteq A^\varepsilon$ for some $\varepsilon,\delta<\beta$ with $\varepsilon+\sqrt{2}\delta\leq (\sqrt{2}-1)\beta$, then $X$ is a strong deformation retract of $A^r$ along the nearest point projection for all $r \in \frac{1}{2} \left(\beta +\varepsilon \pm \sqrt{2(\beta-\delta)^2-(\beta+\varepsilon)^2}\right)$. 
\end{theorem}
\begin{corollary}\label{cor: closed_set_approximation2}
    Let $X\subseteq\mathbb{R}^d$ be a closed Euclidean subset with positive reach and let $0<\beta<\tau(X)$. If $A\subseteq X^\delta$ and $X\subseteq A^\varepsilon$ for some $\varepsilon,\delta<\beta$ with $\varepsilon+\sqrt{2}\delta\leq (\sqrt{2}-1)\beta$, then
    \begin{equation*}
         d_\mathrm{I}(\Rb(A^r,g),\Rb(X,g|_X)) < k(r+\delta),
    \end{equation*}    
    whenever $r \in \frac{1}{2} \left(\beta +\varepsilon \pm \sqrt{2(\beta-\delta)^2-(\beta+\varepsilon)^2}\right)$ and $g:A^r\to\mathbb{R}$ is $k$-Lipschitz. \qed
\end{corollary}

In \cite{wang2018} they consider compact (and thus closed and bounded), differential manifolds $\mathcal{M}\subseteq\mathbb{R}^d$ possibly with boundary $\partial\mathcal{M}$. The same argument holds in this case.
\begin{theorem}[Thm.\ 3.2 in {\cite{wang2018}}]
    Let $\mathcal{M}\subseteq\mathbb{R}^d$ be a compact, differentiable manifold, and $A\subseteq\mathcal{M}$ finite. If $\mathcal{M}\subseteq A^{\varepsilon/2}$ for $\varepsilon<\beta/2<1/2\cdot \min\{\tau_{\mathcal{M}},\tau_{\partial\mathcal{M}}\}$ where the tangent space projection $\varphi_{p,\mathcal{M}}:\mathcal{M}\to T_p\mathcal{M}$ restricted to $B_\beta(p)\cap\mathcal{M}$ is a diffeomorphism onto its image for all $p\in\mathcal{M}$, then $\mathcal{M}$ is a (strong) deformation retract of $A^\varepsilon$ along the nearest point projection.
\end{theorem}

\begin{corollary}\label{cor: closed_set_approximation3}
    Let $\mathcal{M}\subseteq\mathbb{R}^d$ be a compact, differentiable manifold, and $A\subseteq\mathcal{M}$ finite. If $\mathcal{M}\subseteq A^{\varepsilon/2}$ for $\varepsilon<\beta/2<1/2\cdot \min\{\tau_{\mathcal{M}},\tau_{\partial\mathcal{M}}\}$ where the tangent space projection $\varphi_{p,\mathcal{M}}:\mathcal{M}\to T_p\mathcal{M}$ restricted to $B_\beta(p)\cap\mathcal{M}$ is a diffeomorphism onto its image for all $p\in\mathcal{M}$, then 
        $d_\mathrm{I}(\Rb(A^\varepsilon,g),\Rb(X,g|_X))< k\cdot\varepsilon$,
    for any $k$-Lipschitz map $g:A^\varepsilon\to\mathbb{R}$. \qed
\end{corollary}

\subsection{Reeb approximation from closed subsets of Riemannian manifolds} \label{sec: riemann}
We consider approximating the Reeb graph of closed subsets of a Riemannian manifold, following \cite[Sec.\ 5]{attali2024tight}. As before we need the unknown subset and the sample to be well-behaved, in addition we also need the underlying manifold to be manageable. 

Let $\mathcal{N}$ be a Riemannian manifold. The \textbf{sectional curvature} of $\mathcal{N}$ at a point $p$ is the (unique) quantity
\begin{equation*}
    K_\mathcal{N}(p) = \langle \nabla_Y\nabla_X X - \nabla_X\nabla_Y X,Y \rangle
\end{equation*}
where $X$ and $Y$ are orthogonal unit vectors in $T_p\mathcal{N}$ \cite[Sec.\ 6.3.3]{burago2014}. We say that $\lambda\in\mathbb{R}$ \textbf{bounds the sectional curvature of $\mathcal{N}$ from below} if $\lambda\leq K_\mathcal{N}(p)$ for all $p\in\mathcal{N}$. The \textbf{cut locus} $\operatorname{cl}(X)$ of a closed subset $X$ in $\mathcal{N}$ is the set of points in $\mathcal{N}$ with at least two minimal geodesics to point(s) in $X$ \cite[Def.\ 12]{attali2024tight}. The \textbf{cut locus reach} $\tau^{\operatorname{cl}}(X)$ of $X$ in $\mathcal{N}$ is the smallest distance between $X$ and $\operatorname{cl}(X)$ \cite[Def.\ 14]{attali2024tight}, equivalently
\begin{equation*}
    \tau^{\operatorname{cl}}(X)=\sup\{ r \,|\,X^r\cap \operatorname{cl}(X)=\emptyset\},
\end{equation*}
where the thickening $X^r$ uses the metric $d_{\mathcal{N}}$ defined by the minimal length geodesics.

\begin{theorem}[{\cite[Prop.\ 15]{attali2024tight}}]
    Let $X\subseteq \mathcal{N}$ be a closed subset of a ($C^2$) Riemannian manifold such that $\lambda$ bounds its sectional curvature from below, and assume $0<\beta\leq\tau^{\operatorname{cl}}(X)$. If $A\subseteq X^\delta$ and $X\subseteq A^\varepsilon$ for $\varepsilon,\delta <\beta$ such that
    \begin{gather}
    \begin{aligned}\label{eq: Riemannian_reconstruct_demands}
        2 \cos{\sqrt{\lambda}(\beta-\delta)}-\cos{\sqrt{\lambda}(\beta+\varepsilon)}\leq 1 &&\textrm{if }\lambda > 0,\\
        \sqrt{2}(\beta-\delta)-(\beta+\varepsilon)\leq 0  &&\textrm{if }\lambda = 0,\\
        2 \cosh{\sqrt{\lambda}(\beta-\delta)}-\cosh{\sqrt{\lambda}(\beta+\varepsilon)}\geq 1 &&\textrm{if }\lambda<0,    
    \end{aligned}
    \end{gather}
then $X$ is a deformation retract of $A^r$ along the nearest point projection for $r=1/2\cdot(\beta+\varepsilon)$.
\end{theorem}
The homotopy $H:A^r\times [0,1]\to A^r$ moves each point $p$ in $A^r$ along the unique minimal length geodesic $h_\bullet(p):[0,1]\to A^r$ towards its closest point in $X$. Since $h_\bullet$ is a unique and minimal geodesic, then so is the continuous segment $h_\bullet(p)|_{[t,1]}$ for $0\leq t\leq 1$. In particular, the point $h_t(p)$ is sent to $h_1(p)$ by $h_1$, and $H$ is a path deformation retraction.

\begin{corollary}\label{cor: Riemannian_subset_approximation}
    Let $X\subseteq \mathcal{N}$ be a closed subset of a ($C^2$) Riemannian manifold such that $\lambda$ bounds its sectional curvature from below, and assume $0<\beta\leq\tau^{\operatorname{cl}}(X)$. If $A\subseteq X^\delta$ and $X\subseteq A^\varepsilon$ for $\varepsilon,\delta <\beta$ satisfying equation \eqref{eq: Riemannian_reconstruct_demands}, then $d_{\mathrm{I}}(\Rb(A^r,g),\Rb(X,g|_X))<k(r+\delta)$ for $r=1/2\cdot(\beta+\varepsilon)$ and any $k$-Lipschitz function $g:A^r\to\mathbb{R}$. \qed
\end{corollary}
The value for $r$ can be picked in a greater interval (see the extended version of \cite[B.3.1]{attali2024tight}), and our result will still hold.

\section{Computation} \label{sec: computation}

Given a set of points $A\subseteq \R^d$ and a parameter $\varepsilon$, we describe an algorithm to compute the Reeb graph $\mathcal R :=\Rb(A^\varepsilon,f)$, where $A^\varepsilon$ is the $\varepsilon$-thickening of $A$ and $f:A^\varepsilon\to\R$ is an affine function.
Unless $f$ is constant, up to reparameterization of $\R^d$, it is a projection onto one of the $d$ coordinate axes.
\cref{alg:main} takes as input two sets $\I_A$ and $\I_T$ of closed intervals, where $\I_A$ has an interval $I_p$ for each element $p$ of a set $A$, and $\I_T$ has an interval $J = I_{p,q}\sse I_p\cap I_q$ for every element of $T$, which is a set of unordered pairs $\{p,q\}\sse A$.
To apply the algorithm to the union of $\varepsilon$-balls, we let $I_p = f(B_\varepsilon(p))$ and $I_{p,q} = f(B_\varepsilon(p)\cap B_\varepsilon(p'))$ whenever this is nonempty.
Computing the minimum and maximum of the intersection of a pair of balls in constant time results in $O(n^2)$ time to prepare the input to the algorithm by running over all points and pairs of points.
This is faster than $O(n(n+t) \alpha(n))$, which is the complexity of \cref{alg:main} by \cref{thm:algorithm_complexity}, so the running time of the whole procedure is $O(n(n+t) \alpha(n))$. We first illustrate the algorithm with an example.

\begin{example}
Let $A=\{p,q,r,s\}\sse \R^2$ as drawn in \cref{fig:algorithm}.
Let $\varepsilon=1$, and let $f\colon \R^2 \to \R$ be projection onto the $y$-axis.
We first compute the intervals $I_u$ and $I_{u,v}$ for $u\neq v$ in $A$, and sort them by increasing left endpoints. For ties we put intervals of the form $I_u$ before those of the form $I_{u,v}$.
We iterate through these intervals; the figure shows right before we start the iteration of the last interval, $I_{r,s}$.
We have a partition $\Q$ (stored as a list partition-of-reals in \cref{alg:main}) of $\R$, and for every $J\in \Q$, we have stored a partition $\Pp(J)$ of a subset of $A$ with $u$ and $v$ in the same set if we have discovered that their balls intersect in the same connected component of $f^{-1}(x)$.
We make sure that $\Q$ contains as few intervals as possible after every iteration, so if $J$ and $J'$ are consecutive intervals in $\Q$, then $\Pp(J)\neq \Pp(J')$.

When we get to the iteration of $I_{p,r}$ (line 4 in \cref{alg:main}), we run through the intervals in $\Q$ intersecting $I_{p,r}$, colored red in \cref{fig:algorithm}.
For each such $J$, we merge the sets in $\Pp(J)$ containing $r$ and $s$ (line 18, which calls \cref{alg:union}; the merging happens on line 2 in \cref{alg:union}).
In $\{prs\}$, this makes no difference, so we move on to $\{ps,r\}$, which becomes $\{prs\}$ after merging.
Since this is equal to the partition for the previous interval in $\Q$, we combine the two intervals into one to keep $\Q$ minimal (line 10 and 11 in \cref{alg:union}; the latter calls \cref{alg:delete}).
Similarly, the next partition $\{r,s\}$ becomes $\{rs\}$ after merging.
However, this interval only partly overlaps with $I_{p,q}$, so we first need to split it in two (line 5 in \cref{alg:union}, which calls \cref{alg:split}), and then do the merging operation over the interval contained in $I_{p,q}$ (line 6 in \cref{alg:union}).
This gives the partitions on the right in \cref{fig:algorithm}, which encode the Reeb graph next to it.
\end{example}
\begin{figure}[h]
\centering
\begin{tikzpicture}[scale=1.3]
\node[below] at (.1,1){$p$};
\fill (.1,1) circle (.05cm);
\fill[opacity=.3] (.1,1) circle (1cm);
\node[right] at (1.4,-.4){$q$};
\fill (1.4,-.4) circle (.05cm);
\fill[opacity=.3] (1.4,-.4) circle (1cm);
\node[right] at (1.9,1.3){$r$};
\fill (1.9,1.3) circle (.05cm);
\fill[opacity=.3] (1.9,1.3) circle (1cm);
\node[above] at (.5,2){$s$};
\fill (.5,2) circle (.05cm);
\fill[opacity=.3] (.5,2) circle (1cm);
\draw[thick] (3,-1.4) to (3,.6);
\draw[dotted] (1.4,-1.4) to (6,-1.4);
\node[below] at (3,-1.4){$I_q$};
\draw[thick] (3.3,0) to (3.3,2);
\draw[dotted] (.1,0) to (6,0);
\draw[dotted] (.1,2) to (6.24,2);
\node[below] at (3.3,0){$I_p$};
\draw[thick] (3.6,.1) to (3.6,.5);
\draw[dotted] (0.54,.1) to (6,.1);
\draw[dotted] (0.96,.5) to (6,.5);
\node[below] at (3.6,.1){$I_{p,q}$};
\draw[thick] (3.9,.3) to (3.9,2.3);
\draw[dotted] (1.9,2.3) to (6.4,2.3);
\node[below] at (3.9,.3){$I_r$};
\draw[thick] (4.2,.3) to (4.2,.6);
\draw[dotted] (1.9,.3) to (6,.3);
\draw[dotted] (1.4,.6) to (6,.6);
\node[below] at (4.2,.3){$I_{q,r}$};
\draw[thick] (4.5,.75) to (4.5,1.55);
\draw[dotted] (1.07,.75) to (6,.75);
\draw[dotted] (.93, 1.55) to (6,1.55);
\node[below] at (4.5,.75){$I_{p,r}$};
\draw[thick] (4.8,1) to (4.8,3);
\draw[dotted] (.5,3) to (6,3);
\node[below] at (4.8,1){$I_s$};
\draw[thick] (5.1,1) to (5.1,2);
\draw[dotted] (.5,1) to (6,1);
\draw[dotted] (.1,2) to (5.1,2);
\node[below] at (5.1,1){$I_{p,s}$};
\draw[thick,color=red] (5.4,1.09) to (5.4,2.21);
\draw[dotted] (.92,1.09) to (5.4,1.09);
\draw[dotted] (1.48,2.21) to (5.4,2.21);
\node[below] at (5.4,1.09){$I_{r,s}$};
\draw[very thick] (6,-1.4) to (6,.5) to[out=50,in=-50] (6,.75) to (6,3);
\draw[very thick] (6,.5) to[out=130,in=-130] (6,.75);
\draw[very thick] (5.9,0) to (6,.1);
\draw[very thick] (6,1.55) to (6.4,2.3);
\node at (7,-1.1){$\emptyset$};
\draw[dotted] (6,-1.4) to (6.5,-.925) to (7.4,-.925);
\node at (7,-.75){$\{q\}$};
\draw[dotted] (6,0) to (6.5,-.575) to (7.4,-.575);
\node at (7,-.4){$\{p,q\}$};
\draw[dotted] (6,.1) to (6.5,-.225) to (7.4,-.225);
\node at (7,-.05){$\{pq\}$};
\draw[dotted] (6,.3) to (6.5,.125) to (7.4,.125);
\node at (7,.3){$\{pqr\}$};
\draw[dotted] (6,.5) to (6.5,.475) to (7.4,.475);
\node at (7,.65){$\{p,qr\}$};
\draw[dotted] (6,.6) to (6.5,.825) to (7.4,.825);
\node at (7,1){$\{p,r\}$};
\draw[dotted] (6,.75) to (6.5,1.175) to (7.4,1.175);
\node at (7,1.35){$\{pr\}$};
\draw[dotted] (6,1) to (6.5,1.525) to (7.4,1.525);
\node at (7,1.7){\textcolor{red}{$\{prs\}$}};
\draw[dotted] (6,1.55) to (6.5,1.875) to (7.4,1.875);
\node at (7,2.05){\textcolor{red}{$\{ps,r\}$}};
\draw[dotted] (6.24,2) to (6.5,2.225) to (7.4,2.225);
\node at (7,2.4){\textcolor{red}{$\{r,s\}$}};
\draw[dotted] (6.4,2.3) to (6.5,2.575) to (7.4,2.575);
\node at (7,2.75){$\{s\}$};
\draw[dotted] (6,3) to (6.5,2.925) to (7.4,2.925);
\node at (7,3.1){$\emptyset$};
\begin{scope}[xshift=2.5cm]
\draw[very thick] (6,-1.4) to (6,.5) to[out=50,in=-50] (6,.75) to (6,3);
\draw[very thick] (6,.5) to[out=130,in=-130] (6,.75);
\draw[very thick] (5.9,0) to (6,.1);
\draw[very thick] (6,2.21) to (6.1,2.3);
\node at (7,-1.1){$\emptyset$};
\draw[dotted] (6,-1.4) to (6.5,-.925) to (7.4,-.925);
\node at (7,-.75){$\{q\}$};
\draw[dotted] (6,0) to (6.5,-.575) to (7.4,-.575);
\node at (7,-.4){$\{p,q\}$};
\draw[dotted] (6,.1) to (6.5,-.225) to (7.4,-.225);
\node at (7,-.05){$\{pq\}$};
\draw[dotted] (6,.3) to (6.5,.125) to (7.4,.125);
\node at (7,.3){$\{pqr\}$};
\draw[dotted] (6,.5) to (6.5,.475) to (7.4,.475);
\node at (7,.65){$\{p,qr\}$};
\draw[dotted] (6,.6) to (6.5,.825) to (7.4,.825);
\node at (7,1){$\{p,r\}$};
\draw[dotted] (6,.75) to (6.5,1.175) to (7.4,1.175);
\node at (7,1.35){$\{pr\}$};
\draw[dotted] (6,1) to (6.5,1.525) to (7.4,1.525);
\node at (7,1.7){\textcolor{red}{$\{prs\}$}};
\draw[dotted] (6,2) to (6.5,1.875) to (7.4,1.875);
\node at (7,2.05){\textcolor{red}{$\{sr\}$}};
\draw[dotted] (6,2.21) to (6.5,2.225) to (7.4,2.225);
\node at (7,2.4){\textcolor{red}{$\{r,s\}$}};
\draw[dotted] (6.1,2.3) to (6.5,2.575) to (7.4,2.575);
\node at (7,2.75){$\{s\}$};
\draw[dotted] (6,3) to (6.5,2.925) to (7.4,2.925);
\node at (7,3.1){$\emptyset$};
\end{scope}
\end{tikzpicture}
\caption{From left to right: the discs of radius $1$ around the points in $A$, the images of the discs and pairwise intersections, and the constructed Reeb graphs with partitions before and after handling the last interval $I_{r,s}$.
We use the shorthand $\{pq,r\}$ for $\{\{p,q\},\{r\}\}$, etc.}
\label{fig:algorithm}
\end{figure}
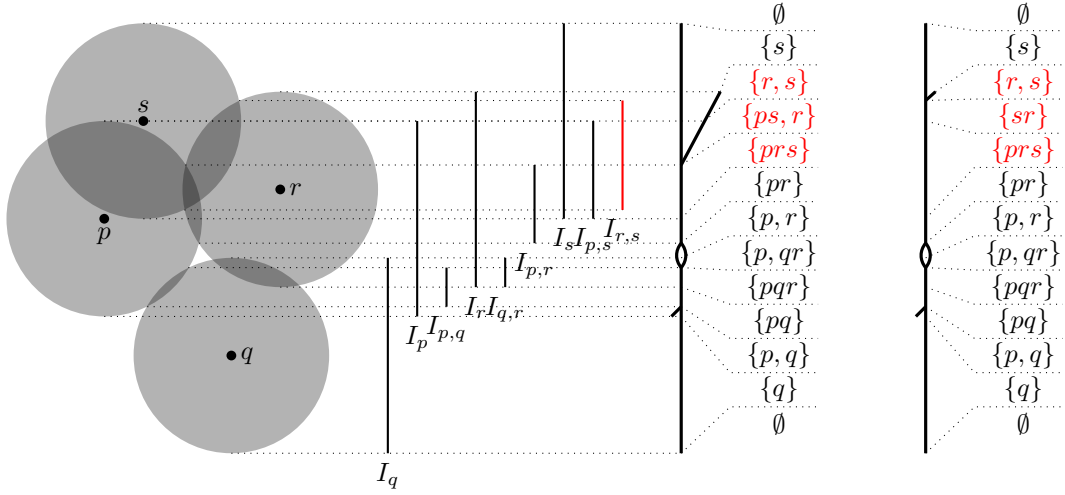

\subsection{The algorithm}

We now describe the algorithm and the data structures in more detail. The main algorithm is shown in \cref{alg:main}, and helper functions in \cref{alg:union}, \cref{alg:delete}, and \cref{alg:split}.
The input to the algorithm is two (finite) sets of closed intervals $\I_A$ and $\I_T$, where $\I_A$ has an interval $I_p$ for each element $p$ of a set $A$, and $\I_T$ has an interval $J = I_{p,q}\sse I_p\cap I_q$ for every element of $T$, which is a set of unordered pairs $\{p,q\}\sse A$.
We assume that the elements of $\I_A$ and $\I_T$ are equipped with the appropriate labels of the form $p$ and $\{p,q\}$, respectively.

\begin{definition}
\label{def:alg_Reeb_graph}
Let $\mathcal R(\I_A,\I_T)$ be the $\R$-space whose underlying space is $\bigsqcup_A \{p\}\times I_p/\sim$, where $\sim$ is generated by the relations $(p,x)\sim (q,x)$ for $x\in I_{p,q}\in \I_T$.
The function $f\colon \mathcal R(\I_A,\I_T)\to \R$ is defined by $(p,x)\mapsto x$.
For $x\in \R$, let $\sim_x$ be the equivalence relation on $\{p \in A \mid x\in I_p\}$ defined by $p\sim_x q$ if $(p,x)\sim (q,x)$, and let $\Pp^{\I_A\cup\I_T}_x$ be the set of equivalence classes of $\sim_x$.
\end{definition}

The algorithm constructs $\mathcal R(\I_{A'},\I_{T'})$ for increasing subsets $A'\sse A$ and $T'\sse T$, with one element being added to either $A'$ or $T'$ for each pass through the for loop at line 4 of \cref{alg:main}.
For two non-intersecting intervals $I$ and $J$, write $I<J$ if $x<y$ for $x\in I$ and $y\in J$.
By interval, we always mean nonempty interval.
Throughout the algorithm, we are updating a data structure built on partition-of-reals, which is a linked list $[J_1 < J_2 < \dots]$ of intervals that form a partition $\Q$ of $\R$.
When we have constructed $\mathcal R(\I_{A'},\I_{T'})$, $\Q$ is the coarsest partition of $\R$ with the property that for every $J\in \Q$ and $x,y\in J$, we have $\Pp^{\I_{A'}\cup\I_{T'}}_x = \Pp^{\I_{A'}\cup\I_{T'}}_y$.
Every entry $J$ has the following attributes:
\begin{itemize}
\item pointers $J$.pred and $J$.succ to its predecessor and successor in the list,
\item a union-find data structure $J$.UF that stores a partition of a subset of $A$.
\end{itemize}
The union-find data structure is a forest of rooted trees, each tree having one set of the partition as its vertex set, and each vertex has a pointer to its parent.
It supports the operations make-set($p$) (creates a new set $\{p\}$ in the partition), union($p,q$) (replaces the sets containing $p$ and $q$ with their union) and find-set($p$) (returns the root of the tree containing $p$) for $p,q\in A$.
For each pair of consecutive intervals $J<J'$ in partition-of-reals, we store a bijective undirected graph $G(J,J')$ on $\Pp_J\sqcup \Pp_{J'}$, where $\Pp_J$ is the partition stored by $J$.UF and similarly for $\Pp_{J'}$. The graph
$G(J,J')$ has an edge between $S\in \Pp_J$ and $T\in \Pp_{J'}$ if $S$ and $T$ intersect, which will be true if and only if $S\sse T$ or $T\sse S$.
This graph is stored as pointers in both directions between $s$ and $t$ whenever there is an edge between $S$ and $T$, where $s$ and $t$ are the roots representing the sets $S$ and $T$ in $J$.UF and $J'$.UF, respectively.
If a root has several neighbors, the pointers are stored in a linked list that is sorted according to a total order of $A$ that is the same for all the graphs.

The graphs $G(J,J')$ encode how the inverse images $f^{-1}(J)$ in $\mathcal R(\I_A,\I_T)$ are connected:
$f^{-1}(J)$ is a disjoint set of copies of $J$ such that the restriction of $f$ to each copy is the identity. 
Let $C_J$ be the set of connected components of $f^{-1}(J)$.
There is an edge in $G(J,J')$ between elements representing $x\in C_J$ and $y\in C_{J'}$ if and only if their closures intersect in $\mathcal R(\I_A,\I_T)$.
The data of $\Q$ together with these bipartite graphs describe $\mathcal R(\I_A,\I_T)$ up to isomorphism of $\R$-spaces, and these intervals and graphs are part of the output of \cref{alg:main}. 

\begin{algorithm}[H]
\caption{main$(\I_A, \I_T)$}
\label{alg:main}
\begin{algorithmic}[1]
\STATE Sort $\I_A\sqcup \I_T$ by increasing left endpoints, with intervals in $\I_A$ before those in $\I_T$ in case of ties.
Let $\I$ be the sorted list.
\STATE partition-of-reals = $[\R]$
\STATE $\R$.UF = $\emptyset$
\FOR{$I\in \I_A\sqcup \I_T$}
    \IF{$I$ is of the form $I_p$}
        \FOR{each interval $J$ in partition-of-reals intersecting $I$}
            \IF{$J\sse I$}
                \STATE $J$.UF.make-set(p)
                \STATE update $G$($J$.pred, $J$) with an edge $p \leftrightarrow p$ if $p$ is in $J$.pred.UF
            \ENDIF
            \IF{$J\nsubseteq I$}
                \STATE $L=J$.split \COMMENT{$L$ is $J\cap I$}
                \STATE $L$.UF.make-set(p)
                \STATE update $G$($L$.pred, $L$) with an edge $p \leftrightarrow p$ if $p$ is in $L$.pred.UF
            \ENDIF
        \ENDFOR
    \ENDIF
    \IF{$I$ is of the form $I_{p,q}$}
        \STATE first = TRUE
        \FOR{each interval $J$ in partition-of-reals intersecting $I$}
            \IF{$J$.UF.find-set(p) $\neq J$.UF.find-set(q)}
                \STATE union($J$,$I_{p,q}$,first)
                \STATE first = FALSE
            \ENDIF
        \ENDFOR
    \ENDIF
\ENDFOR
\end{algorithmic}
\end{algorithm}

\begin{algorithm}[H]
\caption{union($J$,$I_{p,q}$,first)}
\label{alg:union}
\begin{algorithmic}[1]
\IF{$J\sse I$}
    \STATE $J$.UF.union(p,q), which merges trees with roots $r,r'$ to a tree with root $r$
    \STATE in $G(J$.pred$,J)$ and $G(J,J$.succ$)$, add edges between $r$ and all the neighbors $r'$ had before the previous line
\ELSE
    \STATE $J$ = split$(J,I)$ \COMMENT{This changes $J$ to $J\cap I$}
    \STATE $J$.UF.union(p,q), which merges trees with roots $r,r'$ to a tree with root $r$
    \STATE in $G(J$.pred$,J)$ and $G(J,J$.succ$)$, add edges between $r$ and all the neighbors $r'$ had before the previous line
\ENDIF
\IF{first}
    \IF{$J$.UF and $J$.pred.UF encode the same partition}
        \STATE change underlying interval of $J$.pred to $J$.pred $\cup J$
        \STATE delete($J$)
    \ENDIF
\ENDIF
\end{algorithmic}
\end{algorithm}

\begin{algorithm}[H]
\caption{delete($J$)}
\label{alg:delete}
\begin{algorithmic}[1]
\STATE $J$.pred.succ = $J$.succ
\STATE $J$.succ.pred = $J$.pred
\STATE $G$($J$.pred, $J$.succ) = $G$($J$.pred, $J$) $\circ$ $G$($J$, $J$.succ)
\end{algorithmic}
\end{algorithm}

\begin{algorithm}[H]
\caption{split($J,I$)}
\label{alg:split}
\begin{algorithmic}[1]
\FORALL{$p\in J$.UF}
    \STATE $J$.UF.find-set($p$)
\ENDFOR
\STATE $J_1 = J\cap I$
\STATE $J_2 = J\setminus I$
\IF{$J_2$ is connected}
    \IF{$J_1>J_2$}
        \STATE swap $J_1$ and $J_2$
    \ENDIF
    \STATE $J_1$.UF = $J_2$.UF = $J$.UF
    \STATE $J$.pred.succ = $J_2$.pred = $J_1$
    \STATE $J$.succ.pred = $J_1$.succ = $J_2$
    \STATE $G(J_1,J_2)$ = the identity on the roots of $J$
    \STATE $G$($J_1$.pred$,J_1$) = $G$($J_1$.pred$,J$)
    \STATE $G$($J_2$,$J_2$.succ) = $G$($J_2$,$J$.succ)
\ELSE
    \STATE do the same as when $J_2$ is connected, but with three intervals $J_1<J_2<J_3$ replacing $J$ in partition-of-reals
\ENDIF
\RETURN the $J_i$ with $J_i=J\cap I$
\end{algorithmic}
\end{algorithm}

\subsection{Analysis of the algorithm}

We now prove the following lemma, which shows that the algorithm outputs what we expect it to.
For an interval $J$ in partition-of-reals, let $\Pp(J)$ be the partition stored in $J$.UF.
For $p\in A$, let $[p]_J$ be the set in $\Pp(J)$ containing $p$ if such a set exists.
Equivalently, $[p]_J$ is the set of nodes of $J$.UF that are in the same tree as $p$.
(i) shows that $\Pp(J)$ has an element for every connected component of $f^{-1}(J)$, and that these are the explicit representations of the elements of $f^{-1}(x)$ for every $x\in J$ used when defining $\mathcal R(\I_A,\I_T)$ in \cref{def:alg_Reeb_graph}.
(ii) shows that partition-of-reals is minimal, and (iii) shows that the graphs $G(J,J')$ correctly connect the connected components of $f^{-1}(J)$ for different $J$ to construct $\mathcal R(\I_A,\I_T)$.
\begin{lemma}
\label{lem:algorithm}
After the outer for loop in \cref{alg:main} has run through $\I \sse \I_A\cup \I_T$, we have the following:
\begin{itemize}
\item[(i)] For each $J$ in partition-of-reals, $\Pp(J) = \Pp^\I_x$ for every $x\in J$.
\item[(ii)] The intervals in partition-of-reals form the minimal partition of $\R$ such that the previous point holds.
\item[(iii)] In each graph $G(J,J')$ for consecutive elements $J,J'$ of partition-of-reals, there is an edge between a root $x$ in $J$.UF and a root $y$ in $J'$.UF if and only if $[x]_J$ and $[y]_{J'}$ intersect.
\item[(iv)] Let $I$ be the last interval in the list $\I$.
For any $J<J'$ in partition-of-reals with $J\nless I$, and $r\in A$ that is stored in $J'$.UF, we have $[r]_{J'} \sse [r]_J$.
\end{itemize}
\end{lemma}

\begin{proof}[Proof of \cref{lem:algorithm}]
We prove the lemma by induction.
If $\I = \emptyset$, the claims are easy to check, so assume that the claims are true for $\I\setminus \{I\}$.
We will show that they hold also after adding $I$.

First, suppose $I$ is of the form $I_p$ for some $p\in A$.

(i) Note that by the ordering chosen in line 1 of the main algorithm, we do not have $p \sim_x^\I q$ for any $x$.
Thus, to go from $\Pp^{\I\setminus \{I\}}_x$ to $\Pp^\I_x$ for every $x$, we need to split any $J$ that intersects $I$, but is not contained in $I$, and then add $\{p\}$ as a separate set in $\Pp(J)$ for all $J$ in partition-of-reals contained in $I_p$.
This is exactly what the algorithm does.

(ii) The splitting of intervals does not contradict minimality, since if $x\in I_p$ and $y\notin I_p$, then we must have $\Pp^\I_x\neq \Pp^\I_y$ since they are partitions of different sets.
Moreover, if $\Pp(J)$ and $\Pp(J')$ are different partitions not containing $p$, then adding $\{p\}$ to one or both does not make them equal.
Thus, since (ii) holds for the previous iteration by the inductive hypothesis, it holds also after the iteration for $I$.

(iii) We only need to add edges from $p$ to $p$ whenever they appear over consecutive intervals, and the algorithm does this.

(iv) Immediate, because the only thing we do is adding singleton sets to $\Pp(J)$ for $J$ intersecting $I$.

Next, assume $I$ is of the form $I_{p,q}$ for $p,q\in A$.

(i) To go from $\Pp^{\I\setminus \{I\}}_x$ to $\Pp^\I_x$ for every $x$, we need to merge the sets containing $p$ and $q$ in $\Pp(J)$ for every $J\sse I$, potentially after splitting some $J$.
Line 17 in \cref{alg:main} checks if we need to merge, and if it does, it calls \cref{alg:union}, which if necessary calls \cref{alg:split} to split $J$, and then calls UF.union to merge the sets of $p$ and $q$.

(ii) The only thing that can cause non-minimality is if there are consecutive $K$ and $K'$ in partition-of-reals with $\Pp(K)$ and $\Pp(K')$ becoming equal after merging the sets of $p$ and $q$ in $K'$.
Since (iv) is true after the previous iteration, once the condition on line 17 holds, it holds for the remaining $J$ that we iterate over.
Thus, $\Pp(K)$ and $\Pp(K')$ becoming equal can only happen if $K'$ is the first $J$ for which the condition on line 17 holds and $K'$.prev = $K$.
In this case, the variable `first' (line 15) is true, which triggers a check of whether $\Pp(K)=\Pp(K')$ on line 9 in \cref{alg:union}.
If they do, then $K$ and $K'$ are merged, preserving minimality.

(iii) To preserve this property, whenever we merge two sets with roots $r,r'$ in $J$.UF, the neighbors of the root representing the merged set must be the union of the neighbors of $r$ and $r'$ before the merging.
This is taken care of in lines 3 and 7 in \cref{alg:union}.
We also make the necessary changes to the graphs when we split or delete an interval.
Note that we only call \cref{alg:delete} when $\Pp(J)=\Pp(J$.prev$)$, and they have not been changed since before the iteration of $I$ (except potentially splitting the intervals, which does not make any relevant changes to $\Pp(J)$, $\Pp(J$.prev$)$ and $G$($J$.pred, $J$)).
Therefore, $G$($J$.pred, $J$) is a bijection between roots storing equal sets by (iii) from the previous iteration, and it follows that line 3 in \cref{alg:delete} computes the graph $G$($J$.pred, $J$.succ) satisfying (iii).

(iv) By the inductive assumption, (iv) holds for the previous interval $I'$ before we start the iteration of $I$.
Because of the way $\I$ is sorted, the endpoint of $I$ is not to the left of the left endpoint of $I'$, so (iv) holds also for $I$ before the iteration of $I$.
The changes we may make during the iteration are splits and merges of intervals, and taking unions of sets in partitions $\Pp(J)$.
The splits and merges do not affect (iv), since we always split or merge identical partitions.
The unions we take are always of two sets where one contains $p$ and the other $q$, and we take this union for all $J$.UF where $J$ intersects $I$.

For (iv) to fail, we need to take the union of two sets $S$ and $[r]_{J'}$ in $\Pp(J')$ such that $S\nsubseteq [r]_J$.
But then $J'$ has to intersect $I$, so $J$ intersects $I$, and we must have $S=[p]_{J'}$ and $[r]_{J'} = [q]_{J'}$ (or the same with $p$ and $q$ swapped, which is similar).
This means that $S\sse [p]_J$ and $[r]_{J'}\sse [q]_J$ by (iv) from the previous step, so after running $J$.UF.union($p,q$), we get $S\cup [q]_{J'} \sse [p]_J\cup [q]_J = [q]_J =[r]_J$, so (iv) holds.
\end{proof}

We now analyze the computational complexity of the algorithm.
Let $n = |A|$ and $t = |T|$, so $n$ is the number of interval building blocks, and $t$ is the number of gluing intervals.
Let $\alpha(n)$ be the inverse Ackermann function, which is slow-growing to the point that it is constant for all practical purposes.

\begin{theorem}
\label{thm:algorithm_complexity}
\cref{alg:main} runs in $O(n(n+t) \alpha(n))$.
\end{theorem}

\begin{proof}[Proof of \cref{thm:algorithm_complexity}]
Suppose we are about to start the iteration of an $I\in\I$ on line 4 in \cref{alg:main}.
We begin by proving three claims.
\begin{claim}
\label{claim1}
$I$ intersects at most $2n$ intervals in partition-of-reals assuming $n\geq 1$.
\end{claim}

\begin{proof}
By \cref{lem:algorithm} (ii), for any consecutive intervals $J<J'$ in partition-of-reals, $\Pp(J)\neq\Pp(J')$.
They can differ in two ways (or both): (a) they are not partitions of the same set, or (b) there are $r,r'\in A$ that belong to the same set in one partition, but not the other.
Suppose $J$ intersects $I$.
Then, by the sorting of $\I$ by left endpoint, we cannot have that some element $s$ is in $\Pp(J')$, but not in $\Pp(J)$.
Thus, in case (a), $\Pp(J)$ has more points than $\Pp(J')$.
And by \cref{lem:algorithm} (iv), we can only have (b) if $r$ and $r'$ belong to the same set in $J$.UF, but not in $J'$.UF.
For any $K$ in partition-of-reals, let $\omega(K)$ be the number of points times two minus the number of sets in the partition stored by $K$.UF.
By our observations above, in both the cases (a) and (b), we have $\omega(J)<\omega(J')$, because we can go from $J$.UF to $J'$.UF by removing singleton sets and/or splitting sets (at least one such operation, possibly several), both of which decrease $\omega$.
Thus, since $-2n+1\leq \omega(K)\leq 0$ for all $K$ in partition-of-reals, $I$ intersects at most $2n$ intervals in partition-of-reals.
\end{proof}

\begin{claim}
\label{claim2}
For consecutive intervals $J<J'$ in partition-of-reals, $G(J,J')$ has at most $n$ edges.
\end{claim}

\begin{proof}
In what follows, we often view $G(J,J')$ as a graph with $\Pp(J)\sqcup \Pp(J')$ as its set of vertices instead of the roots representing the sets of the partitions.
By \cref{lem:algorithm} (iii), any edge in $G(J,J')$ between $X\in \Pp(J)$ and $Y\in \Pp(J')$ is witnessed by at least one element $p\in X\cap Y$.
Because $\Pp(J)$ and $\Pp(J')$ are partitions, each element $p\in A$ can witness at most one edge, so the claim follows.
\end{proof}

\begin{claim}
\label{claim3}
Suppose $I$ is of the form $I_{p,q}$.
Let $m$ be the number of neighbors in $G(J,J$.succ$)$ of $[p]_J$ summed over all $J$ intersecting $I$.
Then $m\leq 3n-1$.
\end{claim}

\begin{proof}
By \cref{claim1}, there are at most $2n$ edges of the form $[p]_J \leftrightarrow [p]_{J.\text{succ}}$ that count towards $m$.
Any edge not of this form is witnessed by a $r\in A$ with $[p]_J = [r]_J$ and $[p]_{J.\text{succ}} \neq [p]_{J.\text{succ}}$.
By \cref{lem:algorithm} (iv), each $r\in A$ can witness at most one such edge, which means that $m\leq 2n+n-1=3n-1$.
\end{proof}

Line 1 in \cref{alg:main} runs in $O((n+t)\log(n+t)) = O((n+t)\log(n)) = o((n+t)n)$, the first inequality following from $t=O(n^2)$.
We go through the for loop on line 4 in \cref{alg:main} $n+t$ times, once for each interval $I\in \I$.
By \cref{claim1}, we run \cref{alg:union} or UF.make-set at most $2n$ times, which includes up to $2n$ calls to UF.union and two calls to \cref{alg:split} in total, at most once per endpoint of $I$.
It also includes calling \cref{alg:delete} at most once, since the variable `first' is set to FALSE after the first time \cref{alg:delete} is called.
By \cref{claim2} and the fact that $G$($J$.pred, $J$) on line 3 is a bijection, \cref{alg:delete} runs in $O(n)$.
Excluding the $n$ calls to UF.find-set, \cref{alg:split} runs in $O(n)$, since the union find structures and graphs have size $O(n)$.
In \cref{alg:union}, we also update graphs $G(J$.pred$,J)$ and $G(J,J$.succ$)$ on lines 3 and 7.
We do this by merging the sorted linked lists of neighbors of $r$ and $r'$, which we do in $O(\ell)$, where $\ell$ is the sum of the lengths of the two lists.
By \cref{claim3}, the lengths of these lists in $G(J,J$.succ$)$ summed over all $J$ is $O(n)$.
In $G(J$.pred$,J)$, the number of neighbors is $1$ for each of $r$ and $r'$ for all except possibly the first $J$, where the number of neighbors is at most $n$ by \cref{claim2}.
Combining all this, the total length of the lists we need to merge in the iteration of $I$ is $O(n)$, and we can merge these in $O(n)$.
Thus, excluding $O((n+t)n)$ calls to UF.make-set, UF.union and UF.find-set, the algorithm runs in $O((n+t)n)$.

Next, we would like to argue that since UF.make-set, UF.union and UF.find-set run in $\alpha(n)$ time on average, the whole algorithm runs in $O((n+t)n\alpha(n))$.
However, since we juggle many union-find structures instead of working on a single structure throughout the algorithm, we cannot apply this fact directly.

Instead, we will apply ranks and potential functions $\Phi(x)$ for every element $x$ stored in a union-find structure $U$ for a more fine-grained analysis.
These are defined in \cite[Section 21.3 and 21.4]{cormen2022introduction}, where we find the following:
\begin{itemize}
\item[(i)] After make-set($x$), $x$.rank $=0$.
\item[(ii)] After union($x,y$), where $x$.rank $=0$ and $y$.rank is $0$ or $1$, we get $x$.rank $=0$ and $y$.rank $=1$.
\item[(iii)] If $x$ is a root or $x$.rank $=0$, then $\Phi(x) = \alpha(n)\cdot x$.rank.
\item[(iv)] If $x$ is the parent of another node, then $x$.rank $\geq 1$ (Lemma 21.4)
\end{itemize}
The precise values of the ranks and potential functions depend not only on $U$, but also on the sequence of operations used to reach $U$.

The only times we create a new union-find data structure are on line 3 in \cref{alg:main}, which creates a structure storing the empty partition, and in \cref{alg:split} on lines 8 and 9 (and 18, which is similar).
Fix a $J$ and let $U_0$ be $J$.UF when it is initialized, and $U_\infty$ be $J$.UF at the end of the algorithm or right before it is deleted.
In \cref{alg:split}, we initialize $J$.UF right after we have run UF.find-set($p$) for all $p\in A$, which causes all the trees in $U_0$ to have height at most $1$ (that is, all the non-root nodes have a root as their parent).
We can obtain $U_0$ by first calling make-set($p$) for every $p$ stored in the partition, then running union($p,r$) for every root $r$ in $J$.UF and $p$ with $r$ as its parent.
By the above facts (i) and (ii), this causes the roots representing sets of size at least $2$ to have rank $1$ and all other nodes to have rank $0$.
Thus, by fact (iii), the total potential $\Phi(U_0) = \sum_p \Phi(x)$ is $\sigma(U_0)$, which we define as the number of sets stored in $U_0$ of size at least $2$.
Observe that by fact (iv), $\Phi(J$.UF$) \geq \sigma(J$.UF$)$ at any point in the algorithm, and a single operation cannot decrease $\sigma(J$.UF$)$ by more than one.
We now apply \cite[Lemmas 21.11-13]{cormen2022introduction}, which say that the amortized cost of make-set, union and find-set are all at most $\alpha(n)$.
This implies that if we apply these operations $m$ times in total to $J$.UF, the total running time is $O(m\alpha(n)+\Phi(U_0)-\Phi(U_\infty))$.
We have $\Phi(U_0) = \sigma(U_0)$ and $\Phi(U_\infty)\geq \sigma(U_\infty) \geq \sigma(U_0) - m$, so the total runtime is $O(m\alpha(n)+\sigma(U_0)-\sigma(U_0) + m) = O(m\alpha(n))$.
Summing over all $m=n(n+t)$ calls to make-set, union and find-set, we get a cost of $O(n(n+t)\alpha(n))$.

Thus, the runtime of the whole algorithm is $O(n(n+t)(1+\alpha(n))) = O(n(n+t)\alpha(n))$.
\end{proof}

This complexity is in some sense essentially optimal:
Suppose a description of $\mathcal R(\I_A,\I_T)$ has a partition $\Q$ of $\R$ into intervals over which $\mathcal R(\I_A,\I_T)$ is constant, 
together with a set $S(J)$ for every $J\in \Q$ with an element for every connected component of $f^{-1}(J)$.
Then there are examples where $\Q$ has $O(n+t)$ elements, and the average size of $S(J)$ is $O(n)$, which yields a total description of size $\Omega(n(n+t))$.
Thus, apart from a negligible factor of $\alpha(n)$, \cref{thm:algorithm_complexity} cannot be improved without a more efficient way of storing $\mathcal R(\I_A,\I_T)$.

\bibliography{bibliography}

@Article{Amenta1999,
author="Amenta, N.
and Bern, M.",
title="Surface Reconstruction by Voronoi Filtering ",
journal="Discrete {\&} Computational Geometry",
year="1999",
month="Dec",
day="01",
volume="22",
number="4",
pages="481--504",
issn="1432-0444",
doi="10.1007/PL00009475"
}

@InProceedings{attali2024tight,
  author =	{Attali, Dominique and Dal Poz Kou\v{r}imsk\'{a}, Hana and Fillmore, Christopher and Ghosh, Ishika and Lieutier, Andr\'{e} and Stephenson, Elizabeth and Wintraecken, Mathijs},
  title =	{{Tight Bounds for the Learning of Homotopy \`{a} la Niyogi, Smale, and Weinberger for Subsets of Euclidean Spaces and of Riemannian Manifolds}},
  booktitle =	{40th International Symposium on Computational Geometry (SoCG 2024)},
  pages =	{11:1--11:19},
  series =	{Leibniz International Proceedings in Informatics (LIPIcs)},
  ISBN =	{978-3-95977-316-4},
  ISSN =	{1868-8969},
  year =	{2024},
  volume =	{293},
  publisher =	{Schloss Dagstuhl -- Leibniz-Zentrum f{\"u}r Informatik},
  address =	{Dagstuhl, Germany},
  doi =		{10.4230/LIPIcs.SoCG.2024.11}
}

@article{attali2013,
title = {{Vietoris–Rips complexes also provide topologically correct reconstructions of sampled shapes}},
journal = {Computational Geometry},
volume = {46},
number = {4},
pages = {448-465},
year = {2013},
note = {27th Annual Symposium on Computational Geometry (SoCG 2011)},
issn = {0925-7721},
doi = {10.1016/j.comgeo.2012.02.009},
author = {Dominique Attali and André Lieutier and David Salinas}
}

@InProceedings{bauer_2015,
  author =	{Bauer, Ulrich and Munch, Elizabeth and Wang, Yusu},
  title =	{{Strong Equivalence of the Interleaving and Functional Distortion Metrics for Reeb Graphs}},
  booktitle =	{31st International Symposium on Computational Geometry (SoCG 2015)},
  pages =	{461--475},
  series =	{Leibniz International Proceedings in Informatics (LIPIcs)},
  ISBN =	{978-3-939897-83-5},
  ISSN =	{1868-8969},
  year =	{2015},
  volume =	{34},
  publisher =	{Schloss Dagstuhl -- Leibniz-Zentrum f{\"u}r Informatik},
  address =	{Dagstuhl, Germany},
  doi =		{10.4230/LIPIcs.SOCG.2015.461}
}

@inproceedings{bauer_2014,
author = {Bauer, Ulrich and Ge, Xiaoyin and Wang, Yusu},
title = {{Measuring Distance between Reeb Graphs}},
year = {2014},
isbn = {9781450325943},
publisher = {Association for Computing Machinery},
address = {New York, NY, USA},
doi = {10.1145/2582112.2582169},
booktitle = {Proceedings of the Thirtieth Annual Symposium on Computational Geometry},
pages = {464–473},
numpages = {10},
location = {Kyoto, Japan},
series = {SOCG'14}
}

@Article{bauer_2021,
author="Bauer, Ulrich
and Landi, Claudia
and M{\'e}moli, Facundo",
title="{The Reeb Graph Edit Distance is Universal}",
journal="Foundations of Computational Mathematics",
year="2021",
month="Oct",
day="01",
volume="21",
number="5",
pages="1441--1464",
issn="1615-3383",
doi="10.1007/s10208-020-09488-3"
}

@InProceedings{bauer_2022,
  author =	{Bauer, Ulrich and Bjerkevik, H\r{a}vard Bakke and Fluhr, Benedikt},
  title =	{{Quasi-Universality of Reeb Graph Distances}},
  booktitle =	{38th International Symposium on Computational Geometry (SoCG 2022)},
  pages =	{14:1--14:18},
  series =	{Leibniz International Proceedings in Informatics (LIPIcs)},
  ISBN =	{978-3-95977-227-3},
  ISSN =	{1868-8969},
  year =	{2022},
  volume =	{224},
  publisher =	{Schloss Dagstuhl -- Leibniz-Zentrum f{\"u}r Informatik},
  address =	{Dagstuhl, Germany},
  doi =		{10.4230/LIPIcs.SoCG.2022.14}
}

@misc{bollen2022,
      title={{Reeb Graph Metrics from the Ground Up}}, 
      author={Brian Bollen and Erin Chambers and Joshua A. Levine and Elizabeth Munch},
      year={2022},
      eprint={2110.05631},
      archivePrefix={arXiv},
      primaryClass={cs.CG}
}

@article{Bubenik_2014,
   title={{Categorification of Persistent Homology}},
   volume={51},
   ISSN={1432-0444},
   DOI={10.1007/s00454-014-9573-x},
   number={3},
   journal={Discrete and Computational Geometry},
   publisher={Springer Science and Business Media LLC},
   author={Bubenik, Peter and Scott, Jonathan A.},
   year={2014},
   month=jan, 
   pages={600–627} 
}

@InProceedings{carrière2017,
  author =	{Carri\`{e}re, Mathieu and Oudot, Steve},
  title =	{{Local Equivalence and Intrinsic Metrics between Reeb Graphs}},
  booktitle =	{33rd International Symposium on Computational Geometry (SoCG 2017)},
  pages =	{25:1--25:15},
  series =	{Leibniz International Proceedings in Informatics (LIPIcs)},
  ISBN =	{978-3-95977-038-5},
  ISSN =	{1868-8969},
  year =	{2017},
  volume =	{77},
  publisher =	{Schloss Dagstuhl -- Leibniz-Zentrum f{\"u}r Informatik},
  address =	{Dagstuhl, Germany},
  doi =		{10.4230/LIPIcs.SoCG.2017.25}
}

@article{carriére2018,
  title={Statistical analysis and parameter selection for mapper},
  author={Carriere, Mathieu and Michel, Bertrand and Oudot, Steve},
  journal={Journal of Machine Learning Research},
  volume={19},
  number={12},
  pages={1--39},
  year={2018}
}

@Article{Chazal2007,
author="Chazal, Frederic
and Lieutier, Andre",
title="{Stability and Computation of Topological Invariants of Solids in $\mathbb{R}^n$}",
journal="Discrete {\&} Computational Geometry",
year="2007",
month="May",
day="01",
volume="37",
number="4",
pages="601--617",
issn="1432-0444",
doi="10.1007/s00454-007-1309-8"
}

@Article{Chazal2009,
author="Chazal, Fr{\'e}d{\'e}ric
and Cohen-Steiner, David
and Lieutier, Andr{\'e}",
title="{A Sampling Theory for Compact Sets in Euclidean Space}",
journal="Discrete {\&} Computational Geometry",
year="2009",
month="Apr",
day="01",
volume="41",
number="3",
pages="461--479",
issn="1432-0444",
doi="10.1007/s00454-009-9144-8"
}

@Article{Cohen-Steiner2009,
author="Cohen-Steiner, David
and Edelsbrunner, Herbert
and Harer, John",
title="{Extending Persistence Using Poincar{\'e} and Lefschetz Duality}",
journal="Foundations of Computational Mathematics",
year="2009",
month="Feb",
day="01",
volume="9",
number="1",
pages="79--103",
issn="1615-3383",
doi="10.1007/s10208-008-9027-z"
}

@InProceedings{curry2024,
  author =	{Curry, Justin and Mio, Washington and Needham, Tom and Okutan, Osman Berat and Russold, Florian},
  title =	{{Stability and Approximations for Decorated Reeb Spaces}},
  booktitle =	{40th International Symposium on Computational Geometry (SoCG 2024)},
  pages =	{44:1--44:17},
  series =	{Leibniz International Proceedings in Informatics (LIPIcs)},
  ISBN =	{978-3-95977-316-4},
  ISSN =	{1868-8969},
  year =	{2024},
  volume =	{293},
  publisher =	{Schloss Dagstuhl -- Leibniz-Zentrum f{\"u}r Informatik},
  address =	{Dagstuhl, Germany},
  doi =		{10.4230/LIPIcs.SoCG.2024.44}
}

@Article{Dey2013,
author="Dey, Tamal K.
and Wang, Yusu",
title="{Reeb Graphs: Approximation and Persistence}",
journal="Discrete {\&} Computational Geometry",
year="2013",
month="Jan",
day="01",
volume="49",
number="1",
pages="46--73",
issn="1432-0444",
doi="10.1007/s00454-012-9463-z",
url="https://doi.org/10.1007/s00454-012-9463-z"
}

@article{de_Silva_2016,
   title={{Categorified Reeb Graphs}},
   volume={55},
   ISSN={1432-0444},
   DOI={10.1007/s00454-016-9763-9},
   number={4},
   journal={Discrete and Computational Geometry},
   publisher={Springer Science and Business Media LLC},
   author={de Silva, Vin and Munch, Elizabeth and Patel, Amit},
   year={2016},
   month=apr, 
   pages={854–906} 
   }

@Article{Parsa2013,
author="Parsa, Salman",
title="{A Deterministic O(m log m)-Time Algorithm for the Reeb Graph}",
journal="Discrete and Computational Geometry",
year="2013",
month="Jun",
day="01",
volume="49",
number="4",
pages="864--878",
issn="1432-0444",
doi="10.1007/s00454-013-9511-3"
}

@inproceedings{harvey2010,
author = {Harvey, William and Wang, Yusu and Wenger, Rephael},
title = {{A randomized O(m log m) time algorithm for computing Reeb graphs of arbitrary simplicial complexes}},
year = {2010},
isbn = {9781450300162},
publisher = {Association for Computing Machinery},
address = {New York, NY, USA},
doi = {10.1145/1810959.1811005},
booktitle = {Proceedings of the Twenty-Sixth Annual Symposium on Computational Geometry},
pages = {267–276},
numpages = {10},
series = {SoCG '10}
}

@Article{Niyogi2008,
author= {Niyogi, Partha
and Smale, Stephen
and Weinberger, Shmuel},
title="{Finding the Homology of Submanifolds with High Confidence from Random Samples}",
journal="Discrete and Computational Geometry",
year="2008",
month="Mar",
day="01",
volume="39",
number="1",
pages="419--441",
doi="10.1007/s00454-008-9053-2",
url="https://doi.org/10.1007/s00454-008-9053-2"
}

@book{Munkres_2014,
  author    = {Munkres, James},
  title     = {Topology},
  publisher = {Pearson Education Limited},
  year      = {2014},
  address   = {Harlow, Essex},
  edition   = {2nd},
  isbn      = {978-1-292-02362-5}
}

@article{Reeb1946,
  author    = {Reeb, Georges},
  title     = {{Sur les points singuliers d'une forme de Pfaff complètement intégrable ou d'une fonction numérique}},
  journal   = {Comptes Rendus Hebdomadaires des Séances de l’Académie des Sciences},
  volume    = {222},
  pages     = {847--849},
  year      = {1946},
  language  = {French},
  address = {Paris},
  publisher = {Gauthier-Villars}
}

@article{Gelbukh2024,
  author    = {Irina Gelbukh},
  title     = {{On the topology of the Reeb graph}},
  journal   = {Publicationes Mathematicae Debrecen},
  volume    = {104},
  number    = {3-4},
  pages     = {343--365},
  year      = {2024},
}

@misc{Moller_notes,
  author        = {Jesper M. Møller},
  title         = {{General Topology}},
  institution   = {Matematisk Institut, Københavns Universitet},
  howpublished  = {Lecture notes},
  url           = {https://web.math.ku.dk/~moller/e03/3gt/notes/gtnotes.pdf},
  note          = {Accessed: 2025-10-11}
}

@article{wang2018,
title = {Topological inference of manifolds with boundary},
journal = {Computational Geometry},
volume = {88},
pages = {101606},
year = {2020},
issn = {0925-7721},
doi = {https://doi.org/10.1016/j.comgeo.2019.101606},
author = {Yuan Wang and Bei Wang}
}

@article{federer1959,
title={{Curvature Measures}}, 
author={Herbert Federer},
year={1959},
month={12},
journal={Transactions of the American Mathematical Society},
doi = {10.1090/S0002-9947-1959-0110078-1},
pages = {418-491},
volume = {93},
number = {3}
}

@book{burago2014,
  author    = {Dmitri Burago AND Yuri Burago AND Sergei Ivanov},
  title     = {{A Course in Metric Geometry}},
  publisher = {American Mathematical Society},
  year      = {2001},
  address   = {Providence, Rhode Island},
  volume = {33},
  isbn      = {ISBN 0-8218-2129-6},
issn = {1065-7339},
series = {Graduate Studies
in Mathematics}
}

@inproceedings{singh2007Mapper,
    booktitle = {Eurographics Symposium on Point-Based Graphics},
    title = {{Topological Methods for the Analysis of High Dimensional Data Sets and 3D Object Recognition}},
    author = {Singh, Gurjeet and Memoli, Facundo and Carlsson, Gunnar},
    year = {2007},
    publisher = {The Eurographics Association},
    ISSN = {1811-7813},
    ISBN = {978-3-905673-51-7},
    DOI = {/10.2312/SPBG/SPBG07/091-100}
}

@book{dugundji1966,
  author    = {James Dugundji},
  title     = {{Topology}},
  publisher = {Allyn and Bacon, Inc.},
  year      = {1966},
  address   = {470 Atlantic Avenue, Boston},
  isbn      = {0-205-00271-4},
  series = {Allyn and Bacon Series in Advanced Mathematics}
}

@book{cormen2022introduction,
  title={Introduction to algorithms},
  author={Cormen, Thomas H and Leiserson, Charles E and Rivest, Ronald L and Stein, Clifford},
  year={2009},
  edition = {3rd},
  publisher={MIT press}
}

@article{Fasy2022,
author = {Fasy, Brittany Terese and Komendarczyk, Rafal and Majhi, Sushovan and Wenk, Carola},
title = {{On the Reconstruction of Geodesic Subspaces of $\mathbb{R}^n$}},
journal = {International Journal of Computational Geometry \& Applications},
volume = {32},
number = {01n02},
pages = {91-117},
year = {2022},
doi = {10.1142/S0218195922500066}
}

@InProceedings{kim2020,
  author =	{Kim, Jisu and Shin, Jaehyeok and Chazal, Fr\'{e}d\'{e}ric and Rinaldo, Alessandro and Wasserman, Larry},
  title =	{{Homotopy Reconstruction via the Cech Complex and the Vietoris-Rips Complex}},
  booktitle =	{36th International Symposium on Computational Geometry (SoCG 2020)},
  pages =	{54:1--54:19},
  series =	{Leibniz International Proceedings in Informatics (LIPIcs)},
  ISBN =	{978-3-95977-143-6},
  ISSN =	{1868-8969},
  year =	{2020},
  volume =	{164},
  publisher =	{Schloss Dagstuhl -- Leibniz-Zentrum f{\"u}r Informatik},
  address =	{Dagstuhl, Germany},
  doi =		{10.4230/LIPIcs.SoCG.2020.54}
}

@article{Brun2023,
author="Brun, Morten
and Garc{\'i}a Pascual, Bel{\'e}n
and Salbu, Lars M.",
title="Determining homology of an unknown space from a sample",
journal="European Journal of Mathematics",
year="2023",
month="Oct",
day="03",
volume="9",
number="4",
pages="90",
issn="2199-6768",
doi="10.1007/s40879-023-00683-4"
}

@inproceedings{chazal2008towards,
author = {Chazal, Fr\'{e}d\'{e}ric and Oudot, Steve Yann},
title = {Towards persistence-based reconstruction in euclidean spaces},
year = {2008},
isbn = {9781605580715},
publisher = {Association for Computing Machinery},
address = {New York, NY, USA},
doi = {10.1145/1377676.1377719},
booktitle = {Proceedings of the Twenty-Fourth Annual Symposium on Computational Geometry},
pages = {232–241},
numpages = {10},
location = {College Park, MD, USA},
series = {SCG '08}
}

@article{adams2019,
title = {{Metric thickenings of Euclidean submanifolds}},
journal = {Topology and its Applications},
volume = {254},
pages = {69-84},
year = {2019},
issn = {0166-8641},
doi = {https://doi.org/10.1016/j.topol.2018.12.014},
author = {Henry Adams and Joshua Mirth}
}

@inproceedings{Hausmann1996,
title = {{On the Vietoris-Rips complexes and a Cohomology Theory for metric spaces}},
series = {Prospects in Topology (AM-138)},
booktitle = {Proceedings of a Conference in Honor of William Browder},
author = {Jean-Claude Hausmann},
volume= {138},
publisher = {Princeton University Press},
address = {Princeton},
pages = {175--188},
doi = {10.1515/9781400882588-013},
isbn = {9781400882588},
year = {1996}
}

@article{Gromov1978,
author = {Mikhael Gromov},
title = {{Homotopical effects of dilatation}},
volume = {13},
journal = {Journal of Differential Geometry},
number = {3},
publisher = {Lehigh University},
pages = {303 -- 310},
year = {1978},
doi = {10.4310/jdg/1214434601}
}

@article{Shinagawa1991,
  author={Shinagawa, Y. and Kunii, T.L.},
  journal={IEEE Computer Graphics and Applications}, 
  title={{Constructing a Reeb graph automatically from cross sections}}, 
  year={1991},
  volume={11},
  number={6},
  pages={44-51},
  doi={10.1109/38.103393}}

@article{Cole-McLaughlin2004,
author="Cole-McLaughlin, Kree
and Edelsbrunner, Herbert
and Harer, John
and Natarajan, Vijay
and Pascucci, Valerio",
title="{Loops in Reeb Graphs of 2-Manifolds}",
journal="Discrete {\&} Computational Geometry",
year="2004",
month="Jul",
day="01",
volume="32",
number="2",
pages="231--244",
issn="1432-0444",
doi="10.1007/s00454-004-1122-6"
}

@article{biasotti2008,
title = {Reeb graphs for shape analysis and applications},
journal = {Theoretical Computer Science},
volume = {392},
number = {1},
pages = {5-22},
year = {2008},
note = {Computational Algebraic Geometry and Applications},
issn = {0304-3975},
doi = {https://doi.org/10.1016/j.tcs.2007.10.018},
author = {Silvia Biasotti and Daniela Giorgi and Michela Spagnuolo and Bianca Falcidieno}
}

@InProceedings{Shi2014,
author="Shi, Yonggang
and Li, Junning
and Toga, Arthur W.",
title="{Persistent Reeb Graph Matching for Fast Brain Search}",
booktitle="Machine Learning in Medical Imaging",
year="2014",
publisher="Springer International Publishing",
address="Cham",
pages="306--313",
isbn="978-3-319-10581-9"
}

@phdthesis{Shailja2024,
  title        = {Reeb graphs for topological connectomics},
  author       = {S. Shailja},
  year         = 2024,
  month        = {June},
  address      = {Santa Barbara, CA},
  school       = {University of California Santa Barbara},
  type         = {PhD thesis}
}

@Article{Brown2021,
author="Brown, Adam
and Bobrowski, Omer
and Munch, Elizabeth
and Wang, Bei",
title="Probabilistic convergence and stability of random mapper graphs",
journal="Journal of Applied and Computational Topology",
year="2021",
month="Mar",
day="01",
volume="5",
number="1",
pages="99--140",
issn="2367-1734",
doi="10.1007/s41468-020-00063-x"
}

@InProceedings{munch2016,
  author =	{Munch, Elizabeth and Wang, Bei},
  title =	{{Convergence between Categorical Representations of Reeb Space and Mapper}},
  booktitle =	{32nd International Symposium on Computational Geometry (SoCG 2016)},
  pages =	{53:1--53:16},
  series =	{Leibniz International Proceedings in Informatics (LIPIcs)},
  ISBN =	{978-3-95977-009-5},
  ISSN =	{1868-8969},
  year =	{2016},
  volume =	{51},
  publisher =	{Schloss Dagstuhl -- Leibniz-Zentrum f{\"u}r Informatik},
  address =	{Dagstuhl, Germany},
  URN =		{urn:nbn:de:0030-drops-59454},
  doi =		{10.4230/LIPIcs.SoCG.2016.53}
}

\end{document}